\documentclass[10pt]{article}
\usepackage{tikz}
\usepackage{authblk}
\usepackage{graphicx}
\usepackage{amsfonts}
\usepackage{float}
\usepackage{mathtools}

\usepackage{subfigure}

\usepackage{amssymb,latexsym,amsmath,amscd,slashed,mathtools,mathrsfs,bm,stmaryrd,dsfont,amsthm}
\usepackage{setspace,scalerel,color,graphicx,enumerate}

\newcommand{\rr}{\mathbb{R}}
\newcommand{\cc}{\mathbb{C}}
\newcommand{\zz}{\mathbb{Z}}

\newcommand{\hh}{\mathbb{H}}
\newcommand{\dd}{\mathbb{D}}
\newcommand{\mcal}[1]{\mathcal{#1}}
\newcommand{\dif}{\mathrm{d}}

\newcommand{\re}[1]{\text{Re}(#1)}
\newcommand{\im}[1]{\text{Im}(#1)}

\newcommand{\trace}[1]{\text{Tr}(#1)}
\newcommand{\res}{\text{Res}}

\theoremstyle{plain}
\newtheorem{theorem}{Theorem}

\newtheorem{lemma}[theorem]{Lemma}
\newtheorem{prop}[theorem]{Proposition}

\theoremstyle{definition}
\newtheorem{definition}{Definition}
\newtheorem{remark}[definition]{Remark}
\DeclareMathOperator{\tr}{Tr}

\theoremstyle{plain}
\newtheorem{assumption}{Assumption}

\makeatletter

\date
\makeatother

\title{The Schwinger-Dyson equations for random fuzzy geometries coupled to matter}

\usepackage[backend=bibtex]{biblatex}
\author{Jeremy Gamble, Masoud Khalkhali, Nathan Pagliaroli}
\AtEndDocument{%
  \par
  \medskip
  \begin{tabular}{@{}l@{}}%
  \\
    \textsc{Department of Mathematics, University of Western Ontario, Canada}\\
    \textit{E-mail address}: \texttt{jgambl4@uwo.ca}\\
    \\
    \textsc{Department of Mathematics, University of Western Ontario, Canada}\\
    \textit{E-mail address}: \texttt{masoud@uwo.ca}\\
    \\
    \textsc{Department of Combinatorics and Optimization, University of Waterloo}\\
    \textit{E-mail address}: \texttt{npagliar@uwaterloo.ca}\\
  \end{tabular}}

\begin{document}
	\maketitle
	\begin{abstract}
In this work we study the Schwinger-Dyson equations and saddle point equations of matrix integrals that come from type $(0,1)$ random fuzzy geometries coupled to fermions or bosons. Such random fuzzy geometries are bi-tracial Hermitian matrix ensembles with a determinant contribution in the integrand. We derive the Schwinger-Dyson equations using complex analytic techniques from the saddle point equation. For arbitrary potentials with either bosonic or fermionic contributions, their Schwinger-Dyson equations can be solved iteratively.  For both the Gaussian models with either one boson or fermion we rigorously derive the formula for the free energy and first moment in terms of elliptic integrals. In the bosonic case this solution is closely related to the Hoppe model and the three-colour model.
	\end{abstract}
	\tableofcontents
	\section*{Introduction}

A spectral triple $(\mathcal{A},\mathcal{H},D)$ consists of an involutive unital algebra $\mathcal{A}$ of operators on a Hilbert space $\mathcal{H}$, and a self-adjoint Dirac operator $D$ with compact resolvent that acts on $\mathcal{H}$ and is such that the operator norm $||[D,a]||$ is finite for any $a \in \mathcal{A}$. In particular, we are interested in real spectral triples which have an additional anti-linear isometry $J: \mathcal{H} \xrightarrow{} \mathcal{H}$ that is required to satisfy several conditions \cite{varilly2006introduction}.  
The motivation to study real spectral triples is that they serve as noncommutative analogs of   $\text{spin}^c$  Riemannian manifolds. This idea is based on the fact that any closed $\text{spin}^c$ Riemannian manifold $M$ gives rise to a real spectral triple in which the algebra $\mathcal{A}= C^{\infty} (M)$ is the algebra of smooth complex valued functions on  $M$ and the Hilbert space is the space of square integrable sections of the spinor bundle such that the elements of $\mathcal{A}$ act as multiplication operators and the Dirac operator $D$ is the usual Dirac operator of $M$. The metric information of a spin$^{c}$ Riemann manifold can be recovered from its real spectral triple, where the Dirac operator takes the place of the usual metric via Connes' distance formula \cite{connes1992metric}. 

The framework of spectral triples is used in the spectral action principle, which associates the following fundamental action functional with a spectral triple:
\begin{equation*}
   S(D,\psi)= \tr (\chi(D/\Lambda)) + \langle \psi, D \psi\rangle 
\end{equation*}
where $D$ is the Dirac operator, $\chi$ is a smooth even cutoff function, the parameter $\Lambda$ fixes the mass scale, and $\psi$ is a spinor on the Hilbert space. The spectral action was first introduced in the seminal paper \cite{chamseddine1997spectral}, where it was shown that the asymptotic expansion of the spectral action produces the standard model action coupled to Einstein plus Weyl gravity. Since this success, there have been attempts to quantize the spectral action \cite{van2022one,van2023cyclic,gakkhar2024spectral,van2024towards,perez2024spectral,perez2024loop,hekkelman2025power} to arrive at a model of quantum gravity. In particular in \cite{barrett2015matrix}, Barrett considers path integrals over the moduli space of Dirac operators for spectral triples with finite dimensional matrix algebras, known as fuzzy geometries (or fuzzy spectral triples) of the form
\begin{equation*}
    \int e^{-S(D,\psi)}dDd\psi.
\end{equation*}
The moduli space of Dirac operators are finite dimensional real vector spaces of matrices, and thus for an appropriate action are well-defined mathematically. Note that $\psi$ denotes a Grassmann variable, so integration with respect to $\psi$ is Br\'ezin integration. Such integrals are matrix integrals and the associated probability distributions are called Dirac ensembles. Fuzzy geometries and Dirac ensembles have been studied both numerically and analytically. Recent work has found connections with minimal models of conformal field theory \cite{hessam2023double}, topological recursion \cite{azarfar2024random}, and map enumeration \cite{khalkhali2024coloured,khalkhali2022spectral}. Extensions of this framework have also been proposed, such as gauge fuzzy spectral triples \cite{perez2022multimatrix}, BV quantization \cite{gaunt2022bv}, and the addition of internal spaces \cite{barrett2026fuzzy}. See \cite{hessam2022noncommutative} for a recent review.

In this paper we build upon the work in \cite{khalkhali2024large} that studies type $(0,1)$ Dirac ensembles with fermionic contributions to their action. The $(0,1)$ refers to the signature of the associated Clifford module. In this paper we derive the Schwinger-Dyson equations (SDE) for fermionic type $(0,1)$ Dirac ensembles with arbitrary potentials, allowing for a perturbative solution for the moments.  We emphasize because of the fermionic contributions to the integral, deriving the Schwinger-Dyson equations in the standard approach does not produce a system of equations in terms of tracial moments. Thus, a more innovative approach is required. Explicit formulae for the free energy of the Dirac ensemble with a Gaussian potential in terms of elliptic integrals are found. More generally, for any polynomial potential  the SDE can be solved iteratively. We find solutions in terms of elliptic integrals that are deeply related to the Hoppe model \cite{hoppe1989quantum,kazakov1999d} and the three-colour matrix model \cite{di1994entropy,eynard1998iterative,kostov2002exact}.

In Section \ref{sec:background} we provide the necessary background and setup for these models. In Section \ref{sec: DSE from saddle}, we use the saddle point equation to derive the SDE for these models. In Section \ref{sec: Gaussian}, we derive explicit formulae for the second moment, and therefore also the free energy, of the Gaussian models with either bosonic or fermionic matter in terms of elliptic integrals involving the coupling constant.

\section{Background}\label{sec:background}

\subsection{Dirac ensembles}
Fuzzy spectral triples capture the geometric structure of ``fuzzy spaces" such as the fuzzy sphere \cite{madore1992fuzzy} or fuzzy torus \cite{schreivogl2013generalized}. The algebra of smooth functions of a manifold is replaced by an algebra of matrices, the Dirac operator by a self-adjoint matrix, and there is a Clifford algebra action on the spinor space $V$.  
\begin{definition}\cite{barrett2015matrix}
A \textit{fuzzy spectral triple} of signature $(p, q)$ is a real spectral triple of the form
$(M_{N}(\mathbb{C}),V \otimes  M_{N} (\mathbb{C}) , D; J, \Gamma )$
where  $V$ is an irreducible Clifford module for the Clifford algebra $Cl_{p,q}$  with real structure $J_{V}$ and grading $\Gamma_{V}$. The real structure of the fuzzy spectral triple is $J = J_{V}\otimes I$ and if $p + q$
is even, the fuzzy spectral triple is an even real spectral triple with grading $\Gamma =  \Gamma_{V}\otimes I$. Note that $q-p$ is referred to as the $KO$ dimension of the real spectral triple.
\end{definition}

 A \textit{Dirac ensemble}, is a probability distribution on a set of fuzzy spectral triples. Previous works have considered probability distributions on sets of fuzzy spectral triples with a fixed algebra and signature, making it so that the only piece of information that varies is the Dirac operator. For such sets, the moduli space of Dirac operators was classified in \cite{barrett2015matrix}. 

  As discussed in the introduction, recent work \cite{barrett2024,khalkhali2024large,verhoeven2023geometry} has aimed to add fermionic contributions to the partition functions of Dirac ensembles. That is, consider integrals of the form
 \begin{equation}
    Z = \int_{\mathcal{D}}\int_{\mathbb{G}} e^{-\tr S(D) - \langle \psi, D\overline{\psi}\rangle } dD d\psi d\overline{\psi}  . 
\end{equation}
where the fermionic space for this fuzzy spectral triple is taken to be the complex (or optionally real) Grassmann algebra generated by the Hilbert space $V \otimes  M_{N} (\mathbb{C}) $. We restrict our attention in this work to Dirac ensembles for type $(0,1)$ fuzzy geometries: 
$$(M_N(\cc),\cc \otimes M_{N}(\cc),1\otimes \gamma,*\otimes J_V,D_{fuzzy} = [H,\cdot]),$$
where $H$ is some Hermitian matrix.
\begin{remark}
 Note that the space of $N\times N$ Hermitian matrices does not perfectly parameterize the moduli space of Dirac operators here. Since $D$ is expressed as a commutator with a Hermitian matrix $H$, then $\widetilde{H}=H+t I_N$ with $t \in \mathbb{R}$, so that

$$
[\widetilde{H}, m]=\widetilde{H} m-m \widetilde{H}=H m+t m-m H-t m=[H, m] .
$$

Therefore there are many Hermitian matrices $H$ representing the same Dirac operator $D$. The solution used here is to consider for each $D$, the fiber of matrices with form $$\left\{H_0+t I_N: \operatorname{Tr}\left(H_0\right)=0, H_0^*=H_0, D=\left[H_0, \cdot\right] \otimes I_N\right\}$$ which all map to the same Dirac operator $D$. Then on each fiber of $D \in \mathcal{D}_N$ we put the probability measure

$$
\sqrt{\frac{a}{\pi}} \exp \left(-a \operatorname{Tr}(H)^2\right) 
$$
for a new parameter $a$ that is introduced but does not in fact come into any formulae derived. Further details can be found in \cite{khalkhali2024large} and in particular chapter 4.2.2 of \cite{verhoeven2023geometry}. 
\end{remark}

The fermionic integral is going to result in a factor of $\det(D)$ or $\text{pf}({D})$, which, since $D=[H,\cdot]\otimes 1$ has zeroes for eigenvalues, will be zero. The approach taken in \cite{khalkhali2024large}  to avoid such trivialities is to augment the Dirac operator with a mass term while maintaining a spectral triple structure. We now denote our original class of Dirac operators by $D_{fuzzy}$, i.e. $D_{fuzzy} = [H,\cdot]$, and reserve $D$ for Dirac operators  of the new spectral triple. Taking the external product, the spectral triple becomes
\begin{align*}
	(M_N(\cc),M_{N}(\cc)\otimes\cc,1\otimes \gamma,*\otimes J_V,D_{fuzzy})\rightsquigarrow (M_N(\cc),M_{N}(\cc)\otimes \cc^2,1\otimes \sigma_3,*\otimes J, D)
\end{align*}
where 
\begin{align}
	J\begin{bmatrix}
		v_1 \\ v_2 
	\end{bmatrix} = \begin{bmatrix}
	\overline{v_1} \\ -\overline{v_2}
	\end{bmatrix} \qquad D = D_{fuzzy}\otimes \sigma_1 + m\otimes\sigma_2, \label{modified dirac}
\end{align}
$m$ is the fermion mass, and $\sigma_i$ are the Pauli matrices. Note that this does change the spectral triples KO-dimension to zero. Since we have both the normal and reduced fermionic action to consider, we denote by $S_{f}(D,\psi)$ the fermionic action. 

In general, we will choose a polynomial potential of the form 
\begin{equation*}
     \sum_{i=2}^{d}t_{i} \tr D^{i},
\end{equation*}
where $t_{2},...,t_{d}$ are real coupling constants such that the integral converges. For example, quartic $(0,1)$ Dirac ensemble has the space of geometries 
	\begin{align*}
		\mathcal{D}_N = \{ [H,\cdot]\otimes\sigma_1 + m\otimes\sigma_2 : H^*=H \}
	\end{align*}
	with a partition function
	\begin{align*}
		Z = \int_{\mathcal{D}_{N}} \int_{\mathbb{G}}\exp\bigg( -\trace{g_2 D^2+g_4 D^4} \bigg)\exp\bigg( -S_f(D,\psi) \bigg)\dif D \dif\psi\dif\overline{\psi}.
	\end{align*}
    We can now also write our new operator $D$ in terms of $H$. In this example, the bosonic part of the action can be written as 
\begin{align*}
	S(D_H)&=g_4(4N\trace{H^4}-8\trace{H^3}\trace{H}+6\trace{H}^2 \\ &\qquad-8Nm^2\trace{H^2}+8m^2\trace{H}^2-4N^2m^2+Nm^4)\\&\qquad+g_2(4N\trace{H^2}-2\trace{H}^2 + Nm^2).
\end{align*}

We can perform the integral over fermionic variables first. Denote
\begin{align*}
	F(D) = \int_F e^{-S_f(D,\psi)}\dif\psi\dif\overline{\psi}.
\end{align*}
It is not hard to see that $F(D) = \det(D)$ for complex fermionic action and for the real fermionic action  $F(D) = \text{pf}({D})$, both up to a constant factor which can be absorbed into normalization. We now compute the spectrum of $D_{fuzzy}$ and then $D$. Since $H$ is hermitian, it has an orthonormal basis of eigenvectors we denote by $\{v_i\}_{i=1}^N$ with associated eigenvectors $\{\lambda_i\}_{i=1}^n$. Expressing $D_{fuzzy}$ as $H\otimes I_N - I_N\otimes H$, one immediately finds that $\{v_i\otimes v_j\}_{i,j=1}^N$ is a linearly independent set of eigenvectors for $D$ with eigenvalues $\lambda_i-\lambda_j$, thus 
\begin{align*}
	\det(D_{fuzzy}) = \prod_{i,j=1}^N (\lambda_i-\lambda_j).
\end{align*}

With the fermionic part of the integral expressed in terms of the eigenvalues of $H$, and with access to the spectrum of $D$ in terms of $H$, we can now apply Weyl's integration formula  and get a matrix integral purely in terms of the eigenvalues of $H$,
\begin{align*}
	Z &= \int \exp\bigg(-S_{b}(\lambda_{1},...,\lambda_{N})+ \frac{\beta_2}{4}\sum_{i,j}\log(m^2+(\lambda_i-\lambda_j)^2)\\ &\qquad+ \sum_{i\neq j}\frac{\beta}{4}\log((\lambda_i-\lambda_j)^2)-a\bigg(\sum_i \lambda_i\bigg)^2 \bigg)\dif\lambda_1\cdots\dif\lambda_N.
\end{align*}

Recall that the fermionic action is motivated by the noncommutative geometric approach to the standard model. However, we can just as easily introduce bosonic matter by taking the complex or real algebra generated by our Hilbert space. The resulting matter integrals are multivariate Gaussian integrals and the above equation can be reproduced with a change in sign for $\beta_{2}$. We will see that such models are very closely related to known matrix models and to string theory \cite{hoppe1989quantum,eynard1998iterative,kazakov1999d}, which motivates the use of bosonic action.

\subsection{The equilibrium measure}
In random matrix theory, we are often interested in the limiting behaviour of the model as $N\to\infty$. The spectral density or at least its moments, i.e. a probability measure $\mu_E$ such that
\begin{align*}
	\lim_{N\to\infty}\frac{1}{N}\mathbb{E}[ \trace{f(H)}] = \int_\rr f(x)d\mu_E(x)
\end{align*}
will be our focus. 
In \cite{khalkhali2024large, verhoeven2023geometry}, this problem is studied in great detail. The first important thing they show is that, indeed, such a measure $\mu_E$ exists and has compact support, so long as it satisfies certain assumptions. We make the following assumption on the choice of the potential for the proceeding analysis.

\begin{assumption}
    The choice of $S$ is such that the equilibrium measure $\mu_E = \rho(x)\dif x$ for a continuous density function $\rho:\Sigma\to\rr$, where $\Sigma\subset \rr$ is a disjoint union of intervals.
\end{assumption} 

This is generally the case for Hermitian matrix integrals with polynomial potentials, and it is proven to be the case for the type $(0,1)$ model with a fermion and Gaussian or quartic potential in \cite{khalkhali2024large}.
 
 We denote the \textit{$n$th moment} of $\mu_E$ as
\begin{align*}
	\mu_n \coloneqq \int_\Sigma x^n\rho(x)\dif x.
\end{align*}
The measure $\mu_E$ is found by minimizing this functional 
\begin{align*}
	I(\mu) &= \int\bigg(V'(x)-\frac{\beta_2}{4}\log(m^2+(x-y)^2)\\&\qquad-\frac{\beta}{4}\log((x-y)^2)+\frac{a}{2}xy\bigg)\dif\mu(x)\dif\mu(y).
\end{align*}
where $V'(x)$ denotes a polynomial in $x$ whose coefficients are polynomials in $m,t_{2},...,t_{d}$, and moments $\mu_{n}$ that comes from our choice for $S_{b}$.
Further, it is shown that this problem may be transformed into solving the saddle point equation
\begin{align*}
	\text{P.V.} \int_\Sigma \frac{\rho(y)}{y-x}\dif y = -\frac{2}{\beta}V'(x)-\frac{\beta_2}{2}\int_\Sigma \frac{x-y}{m^2+(x-y)^2}\rho(y)\dif y.
\end{align*}
This will be the starting point for our analysis. The authors of \cite{khalkhali2024large} go on to further recast this problem using the Stieltjes transform and the Sokhotski-Plemelj formula. They are able to explicitly solve for $\rho$ when $m=0$ in both the quartic and quadratic case, and obtain some numerical results when $m>0$. It is our goal to study this problem analytically from two different approaches. Firstly, to obtain SDE following the approach of \cite{eynard1998iterative} which allows one to perturbatively solve the problem. Secondly, in the quadratic case ($g_4 = 0, g_2>0$), we use complex analytic tools to study the problem in line with the work done in \cite{hoppe1989quantum,kazakov1999d,kostov2002exact}.

\section{Saddle Point Equation}\label{sec: DSE from saddle}
In this section, we use the saddle point equation for type $(0,1)$ Dirac ensembles with a fermion or boson and a general potential $S$ to derive the SDE. The usual approach to deriving the SDE for Hermitian matrix integrals is to consider the equality 
\begin{equation*}
    \sum_{i,j=1}^{N}\int_{\mathcal{H}_{N}} (H^{\ell})_{i,j} e^{-S(H)}dH = 0 
\end{equation*}
for $\ell>0$, which follows from Stokes' Theorem. By expanding the left hand side one arrives at a recursive equation involving tracial moments for each $\ell$. The obstacle of applying this method here is that the determinant in the integrand of our Dirac ensemble prevents the end result from being in terms of tracial moments.  

We start by considering the resolvent function for this matrix integral,
\begin{align*}
	W(x) \coloneqq \int_\Sigma \frac{\rho(y)}{x-y}\dif y. 
\end{align*}
This function is analytic in the entire complex plane, except for the compact interval(s) $\Sigma$, where it has continuous limits from above and below, which differ. We introduce the notation $W(x\pm i0)$ which is shorthand for $\lim_{\epsilon\to 0} W(x\pm i\epsilon)$. We will often expand $W(x)$ for $|x|>L$, where $\text{Supp}{\rho}\subset [-L,L]$,
\begin{align*}
	W(x) = \int_\Sigma \frac{1}{x}\frac{\rho(y)}{1-y/x}\dif y = \frac{1}{x}\int_\Sigma \rho(y)\bigg(\sum_{i=0}^\infty \frac{y^n}{x^n} \bigg)= \sum_{n=0}^\infty \frac{\mu_n}{x^{n+1}}.
\end{align*}
Note that $\mu_0 = 1$ so that $\rho$ is a valid probability density. \\

We will express the saddle point equation 
\begin{align}
	\text{P.V.}\int_\Sigma \frac{\rho(y)}{y-x}\dif y = -\frac{2}{\beta}V'(x)-\frac{\beta_2}{2}\int_\Sigma \frac{x-y}{m^2+(x-y)^2}\rho(y)\dif y \label{vareqn}
\end{align}
in a different form to derive the SDE. We begin by rewriting the principal value as \cite{verhoeven2023geometry} lemma 3.4.2
\begin{align*}
	\text{P.V.}\int_\Sigma \frac{\rho(y)}{y-x}\dif y = -\frac{1}{2}(W(x+i0)+W(x-i0)),
\end{align*}
valid for $x\in\Sigma$. Next, we use partial fractions to write
\begin{align*}
	\int_\Sigma \frac{x-y}{m^2+(x-y)^2}\rho(y)\dif y &= \frac{1}{2}\bigg(\int_\Sigma \frac{\rho(y)\dif y}{x+mi-y}+\int_\Sigma \frac{\rho(y)\dif y}{x-mi - y}\bigg) = \frac{1}{2}\bigg(W(x+mi)+W(x-mi)\bigg).
\end{align*}
 Hence, we express equation \eqref{vareqn} as the following, for all $x\in \Sigma$:
\begin{align}
	W(x+i0)+W(x-i0)=\frac{4}{\beta}V'(x) +\frac{\beta_2}{2}(W(x+mi)+W(x-mi)) \label{saddle}.
\end{align}

We began by defining the resolvent $W$ as an analytic function on $\cc\setminus\Sigma$, however, it can also be seen as an analytic function on a Riemann surface with infinitely many sheets. The above saddle point equation then describes the behavior of $W$ in successive sheets. For instance, let us rewrite the above equation as 
\begin{align*}
	W(x-i0) = \frac{4}{\beta}V'(x)+\frac{\beta_2}{2}[W(x+mi)+W(x-mi)]-W(x+i0).
\end{align*}
Now, let us temporarily label $W$ with the subscript $k$ indicating which sheet we are on, with $W_0$ being the original sheet, and with the convention that passing through the cut from above leads onto successive sheets, and passing through the cut from below leads to previous sheets. That is, $W_k(x+i0)=W_{k+1}(x-i0)$. We may then write the above equation as
\begin{align*}
	W_1(x+i0) = \frac{4}{\beta}V'(x)+\frac{\beta_2}{2}[W_0(x+mi)+W_0(x-mi)]-W_0(x+i0)
\end{align*}
and we see that as we move into successive sheets, there are additional cuts along which $W_k$ fails to be analytic due to the terms $W(x+mi)+W(x-mi)$, specifically on $\Sigma \pm kmi$ for $k\in \zz$. The existence of these additional cuts suggests a symmetry of the function $W$ which can be exploited to gain additional information about $W$. Namely, we will extract the SDE for our matrix integral by constructing an entire function $F$, and we will prove that it is also constant, allowing us to extract the SDE from the coefficients of $F$ as a linear combination of $\zeta$-functions. This largely follows the strategy used in \cite{eynard1998iterative}, albeit our approach being rigorous.

\subsection{Construction of the Entire Function $H$}

Let us start by introducing some notation. First, $I_k \coloneqq \Sigma+kmi$ and $x_k \coloneqq x-kmi$. This means that the saddle point equation \ref{saddle} may be written for $x\in I_k$ as (noting that $x_k \pm mi = x -(k\mp1)mi=x_{k\mp 1}$)
\begin{align*}
	W(x_k + i0) - W(x_k - i0) = \frac{4}{\beta}V'(x)+\frac{\beta_2}{2}[W(x_{k+1})+W(x_{k-1})].
\end{align*}
Now, multiply this equation by $W(x_k + i0) - W(x_k - i0)$ to obtain
\begin{align*}
	W(x_k +i0)^2 &- W(x_k-i0)^2\\ &= \bigg\{ \frac{4}{\beta}V'(x)+\frac{\beta_2}{2}[W(x_{k+1})+W(x_{k-1})] \bigg\}(W(x_k+i0)-W(x_k-i0)) 
\end{align*}
and rearranging yields
\begin{align*}
	W(x_k+i0)^2-\bigg\{\frac{4}{\beta}V'(x)+\frac{\beta_2}{2}[W(x_{k+1})+W(x_{k-1})]\bigg\}W(x+i0) \\ = W(x_k-i0)^2 - \bigg\{\frac{4}{\beta}V'(x)+\frac{\beta_2}{2}[W(x_{k+1})+W(x_{k-1})]\bigg\}W(x_k-i0)
\end{align*}
implying that the function 
\begin{align*}
	g_k(x) \coloneqq W(x_k)^2 - \bigg\{\frac{4}{\beta}V'(x_k)+\frac{\beta_2}{2}[W(x_{k+1})+W(x_{k-1})]\bigg\}W(x_k)
\end{align*}
is analytic along $I_k$. This suggests a candidate function, built from $W$, that is entire. We define
\begin{align*}
	H(x) := \sum_{k=-\infty}^\infty\underbrace{ W(x_k)^2 - \frac{\beta_2}{2}W(x_k)W(x_{k+1})-\frac{4}{\beta}[V'(x_k)W(x_k)-R(x_k)]}_{\coloneqq f_k(x)}
\end{align*}
where $R(x)$ is the polynomial part of $V'(x)W(x)$. This term must be subtracted in order that $H$ converges.

\begin{prop}
	The function $H$ is entire. 
\end{prop}

\begin{proof}
	Before we prove that $F$ converges to an entire function, we first show that this function cannot have cuts along any $I_k$. Indeed, the function $f_k(x)$ has cuts along $I_k$ and $I_{k+1}$, however, observe
	\begin{align*}
		f_{k-1}(x)+f_k(x)+f_{k+1}(x) &=  g_k(x) + W(x_{k-1})^2 + W(x_{k+1})^2 -\frac{\beta_2}{2}W(x_{k+1})W(x_{k+2}) \\
		&\qquad - \frac{4}{\beta}[ V'(x_{k-1})W(x_{k-1})+V'(x_{k+1})W(x_{k+1})\\
		&\qquad - R(x_k)-R(x_{k-1})-R(x_{k+1}) ].
	\end{align*}
	Thus, $f_{k-1}+f_k + f_{k+1}$ is a function without cuts along $I_k$, since $g_k$ has no cuts along $I_k$, and $R$ is polynomial and therefore has no cuts at all. Considering that $f_k$ may only have cuts along $I_{k-1},I_{k},I_{k+1}$, and $F$ may be written as 
	\begin{align*}
		H(x) = f_{k}(x)+f_k(x)+f_{k+1}(x) + \sum_{j\neq k-1,k,k+1} f_j(x),
	\end{align*}
	i.e. a sum of two functions, neither of which has a cut along $I_k$. Since the choice of $k$ was arbitrary, we conclude that $S$ does not have cuts along any $I_k$. \\

	The previous argument immediately shows that $\sum_{|k|\leq N} f_k(x)$ has no cuts until $I_{\pm N}$, which will be useful in proving the convergence of $F$, which we turn to now. We will show that $H$ converges uniformly on all open balls $B(0;R) \coloneqq \{ x\in \cc : |x|<R \}$. \\
	
	We begin with some estimates that will prove useful. Recall that $L$ is the length of $\Sigma$. This gives a bound on the moments,
	\begin{align*}
		|\mu_n| \leq  \int_\Sigma \bigg| x^n \rho(x)\bigg| \dif x \leq L^n \int_\Sigma \rho(x)\dif x = L^n.
	\end{align*}
	
	Let us now assume $|x|<R$. Taking a term $W(x_k)$, we can expand this when 
	\[ L<|k|m-R<|k|m-|x| \leq \big||x|-|kmi|\big|\leq |x_k|  \]
	i.e. when $|k|>(R+L)/m$, in which case we have
	\begin{align*}
		|W(x_k)|&=\bigg|\sum_{n=0}^\infty \frac{\mu_n}{x_k^{n+1}}\bigg| \leq \frac{1}{|x_k|}\sum_{n=0}^\infty \bigg(\frac{L}{|x_k|}\bigg)^{n}  = \frac{1}{|x_k|-L} \leq \frac{1}{|k|m-(R+L)}.
	\end{align*}
	We now examine the term $V'(x_k)W(x_k)-R(x_k)$ carefully. We will write down the potential in the general form $V'(x)=\sum_{j=0}^{\deg V'}g_{j+1}x^j$. First,
	\begin{align}
		V'(x_k)W(x_k)-R(x_k) &= \sum_{n=0}^\infty\sum_{j=0}^n \frac{g_{j+1}\mu_n}{x_{k}^{n+1-j}} = \underbrace{\sum_{n=0}^\infty \sum_{j=0}^{n-1} \frac{g_{j+1}\mu_n}{x_k^{n+1-j}}}_{A}+\underbrace{\sum_{n=0}^{\deg V'} \frac{g_{n+1}\mu_n}{x_k}}_{B} \label{cotdecomp}
	\end{align}
	The term $A$ may be bounded in a similar manner,
	\begin{align*}
		|\text{(a)}| &\leq \sum_{j=0}^{\deg V'}\frac{1}{|x_k|^{1-j}}\sum_{n=j+1}^\infty \frac{g_{j+1}L^n}{|x_k|^{n}} = \sum_{j=0}^{\deg V'} \frac{g_{j+1}L^{j+1}}{|x_k|(|x_k|-L)}\leq G_1 \frac{1}{(|k|m-R)(|k|m-R-L)}
	\end{align*}
	where $G_1 = \sum_{j=0}^{\deg V'} g_{j+1}L^{j+1}$. For the term $B$, set $G_2 = \sum_{n=0}^{\deg V'}g_{n+1}\mu_n$, so that
	\[ B = G_2 \frac{1}{x_k}. \] 
	
	Our next step is to apply the M-test to $H(x)$ on the ball $B(0;R)$. Set $N=\lceil(R+L)/m\rceil$. We split $S$ into two terms,
	\[ H(x) = \sum_{|k|\leq N} f_k(x) + \sum_{|k|\geq N} f_k(x) \]
	The first term may only have cuts at $I_{\pm N}$ as previously discussed, and $I_{\pm N} = \Sigma \pm \lceil R+L\rceil i$, thus the cuts lie outside of $B(0;R)$, hence the first term $\sum_{|k|\leq N} f_k(x)$ is analytic on $B(0;R)$. For the second term, by our choice of $N$, we may expand each $f_k$ and use our previous estimates to find 
	\begin{align}
		|f_k(x)| &\leq |W(x_k)|^2 + \frac{\beta_2}{2}|W(x_k)||W(x_{k+1})| + \frac{4}{\beta}|V'(x_k)W(x_k)-R(x_k)| \nonumber\\
		&\leq \frac{1}{(|k|m-(R+L))^2}+\frac{\beta_2}{2}\frac{1}{(|k|m-(R+L))(|k+1|m-(R+L))}\nonumber\\&\qquad+\frac{4}{\beta}G_1\frac{1}{(|k|m-R)(|k|m-R-L)} + \frac{4}{\beta}G_2 \frac{1}{|x_k|}. \label{fbound}
	\end{align} 
	Upon summing over $k$, all the terms, except for the last one, clearly converge. However, note that
	\[ \sum_{k\in\zz} \frac{1}{x_k} = \frac{1}{mi}\sum_{k\in\zz} \frac{1}{x/mi-k} = \frac{\pi}{mi}\cot\frac{\pi x}{mi} \]
	by the well known partial fraction expansion of $\cot$. Hence $\sum_{|k|\geq N} \frac{1}{x_k}$ is analytic on $B(0;R)$. Thus, we may apply the M-test to $\sum_{k\in\zz} f_k$ for any $B(0;R)$, hence showing that $\sum_{k\in\zz} f_k$ uniformly converges to a continuous function on any compact subset of $\cc$, and therefore to a holomorphic function on $\cc$ by an application of Morera's theorem.
\end{proof}

One more crucial property to note about the function $H$ we have constructed is that it is periodic i.e. $H(x+mi)=H(x)$. Indeed, since $(x+mi)_k = x+mi - kmi = x-(k-1)mi = x_{k-1}$, we have
\begin{align*}
	H(x+mi)= \sum_{k\in\zz} W(x_{k-1})^2 - \frac{\beta_2}{2}W(x_{k-1})W(x_k) - \frac{4}{\beta}[V'(x_{k-1})W(x_{k-1})-R(x_{k-1})] 
\end{align*}
which is equal to $H(x)$ by a shift $k\mapsto k+1$ in the summation. 

\begin{prop}
	The function $H$ is constant. 
\end{prop}

\begin{proof}
	 Since $H$ is periodic, we need only show that  it is bounded on the infinite strip \[D\coloneqq\{x\in\cc : 0 \leq \im{x} \leq m \}.\]
	We will again make use of the bounds from the convergence of $H$. On the subset of $D$ satisfying $|\re{x}|\leq L$, $H$ is bounded as the set is compact and $F$ analytic. In the complement, i.e. for $|\re{x}|>L$, we have $|x_k|>L$, hence we can expand $f_k$ for all $k$ in this region. We write $H$ as 
	\[H(x) = \frac{G_2 \pi}{mi}\cot\frac{\pi x}{mi}+ \sum_{k\in\zz}\widetilde{f}_k(x)\]
	using the decomposition in \ref{cotdecomp}. It should be understood that $\widetilde{f}$ consists of the terms not contributing to the cotangent function. Writing 
	\[ |H(x)| \leq \bigg| \sum_{k\in\zz}\widetilde{f}_k(x)  \bigg| + \frac{G_2 \pi}{m}\bigg|\cot\frac{\pi x}{mi}\bigg|,\]
	we find that on $\{ x\in D : \re{x}>L \}$,  the first term is bounded due to \ref{fbound}, as is the cotangent term. Hence $H$ is bounded, and must be constant. 
\end{proof}

\subsection{Deriving the Schwinger-Dyson equations}

We are now in a position to derive the SDE using the function $H$. To this end, we will express $H$ as an infinite linear combination of independent $\zeta$-functions. Since we have shown that $H$ is constant, which implies the $\zeta$-function coefficients must vanish, from which the SDE may be derived. \\

For each $a\in \zz_{>0}$, we define
\[ \zeta_a(x) = \sum_{k\in\zz} \frac{1}{x_k^a} = \sum_{k\in\zz} \frac{1}{(x-kmi)^a} = \frac{1}{(mi)^a}\sum_{k\in\zz} \frac{1}{(\frac{x}{mi}-k)^a}. \] 
For notational simplicity, we take $w = x/mi$ and $w_k = w - k$, thus
\[\zeta_a(w) = \frac{1}{(mi)^a}\sum_{k\in\zz} \frac{1}{(w-k)^a} = \frac{1}{(mi)^a}\sum_{k\in\zz} \frac{1}{w_k^a}. \]
We also similarly define
\[ \zeta^\pm_{a,b}(x)= \sum_{k\in\zz}\frac{1}{x_k^a}\frac{1}{(x_k \pm mi)^b}=  \frac{1}{(mi)^{a+b}}\sum_{k\in\zz} \frac{1}{w_k^a}\frac{1}{(w_k \pm 1)^b}=\zeta^\pm_{a,b}(w). \]
Note that $\zeta_1(x)=\frac{\pi}{mi}\cot\frac{\pi x}{mi}$ and that $\zeta_{a+1}(x)=-\frac{1}{a}\zeta'_a(x)$. \\

Let us now rewrite $H$ in terms of these functions, using the usual expansion of $W$
\begin{align*}
	H(x) &= \sum_{k\in\zz} \underbrace{W(x_k)^2}_{\coloneqq A} - \frac{\beta_2}{4}\underbrace{( W(x_k)W(x_{k-1})+W(x_k)W(x_{k+1}))}_{\coloneqq B} - \frac{4}{\beta}\underbrace{[V'(x_k)W(x_k)-R(x_k)]}_{\coloneqq C}.
\end{align*}
Note that we split the $B$ terms into $W(x_k)W(x_{k-1})+W(x_k)W(x_{k+1})$ from $W(x_k)W(x_{k+1})$ by re-indexing. Now we turn to the sum over $A$ terms,
\begin{align*}
	A &= \sum_{k\in\zz}\sum_{a=0}^\infty\sum_{b=0}^\infty \frac{\mu_a \mu_b}{x_k^{a+b+2}} = \sum_{a=0}^\infty\sum_{b=0}^\infty \mu_a\mu_b \zeta_{a+b+2}(x) = \sum_{n=1}^\infty\bigg( \sum_{b=0}^{n-2} \mu_b \mu_{n-b-2} \bigg)\zeta_n(x).
\end{align*}
Next, the $B$ terms,
\begin{align*}
	B &= \sum_{k\in\zz}\sum_{a=0}^\infty\sum_{b=0}^\infty \frac{\mu_a}{x_k^a}\bigg( \frac{\mu_b}{x_{k-1}^b} + \frac{\mu_b}{x_{k+1}^b} \bigg) = \sum_{k\in\zz}\sum_{a=0}^\infty\sum_{b=0}^\infty \frac{\mu_a}{x_k^{a+1}}\bigg( \frac{\mu_b}{(x_{k}+mi)^{b+1}} + \frac{\mu_b}{(x_{k}-mi)^{b+1}} \bigg) \\
	&= \sum_{a=0}^\infty\sum_{b=0}^\infty \mu_a \mu_b (\zeta^+_{a+1,b+1}(x) + \zeta^-_{a+1,b+1}(x)).
\end{align*}
Finally, the $C$ terms,
\begin{align*}
	C &= \sum_{k\in\zz}\sum_{a=0}^\infty\sum_{j=0}^{l} \frac{\mu_a g_{j+1}}{x_k^{a+1-j}} = \sum_{a=0}^\infty\sum_{j=0}^{a} \mu_a g_{j+1} \zeta_{a+1-j}(x) = \sum_{n=1}^\infty\bigg( \sum_{j=0}^{\deg V'} \mu_{n+j-1}g_{j+1} \bigg)\zeta_n(x).
\end{align*}
We further need to express $\zeta_{a,b}^\pm$ in terms of $\zeta_c$ in order to equate coefficients. First, we write
\begin{align*}
	\zeta^\pm_{a,b}(w) = \frac{(-1)^b}{(mi)^{a+b}}\sum_{k\in\zz} \frac{1}{w_k^a}\frac{1}{(\mp 1 - w_k)^b}.
\end{align*}
Now we will make use of partial fraction decomposition, to get the following lemma.

\begin{lemma}
	The functions $\zeta^{\pm}_{a,b}$ can be expressed in terms of $\zeta_l$ as follows:
	\begin{align*}
		\zeta^+_{a,b}(w) &= \frac{(-1)^b}{(mi)^{a+b}}\bigg\{\sum_{l=1}^a {a+b -1 - l\choose b-1}(-1)^{a+b+l}(mi)^l \zeta_l(w)+ \sum_{l=1}^b {a+b-1-l \choose a-1}(-1)^{a+b+l}(mi)^l\zeta_l(w)\bigg\} \\
		\zeta^-_{a,b}(w) &= \frac{(-1)^b}{(mi)^{a+b}}\bigg\{\sum_{l=1}^a {a+b -1 - l\choose b-1}(mi)^l \zeta_l(w)+ \sum_{l=1}^b {a+b-1-l \choose a-1}(mi)^l\zeta_l(w)\bigg\}
	\end{align*}.
\end{lemma}

From this form, we can write
\begin{align*}
	\zeta^+_{a+1,b+1}+\zeta^-_{a+1,b+1} &= \frac{(-1)^{b+1}}{(im)^{a+b+2}}\bigg( \sum_{l=1}^{a+1} {a+b+1-l\choose b}(mi)^l[(-1)^{a+b+l}+1]\zeta_l \\ &\qquad+ \sum_{l=1}^{b+1}{a+b+1-l\choose a}(mi)^l[(-1)^{a+b+l}+1]\zeta_l\bigg)
\end{align*}
and so
\begin{align*}
	\text{(b)}&= \sum_{a,b=0}^\infty \mu_a\mu_b \frac{(-1)^{b+1}}{(mi)^{a+b+2}}\sum_{l=1}^{a+1}{a+b+1-l\choose b}(mi)^l[(-1)^{a+b+l}+1]\zeta_l \\ &\qquad+ \sum_{a,b=0}^\infty \mu_a\mu_b \frac{(-1)^{b+1}}{(mi)^{a+b+2}}\sum_{l=1}^{b+1}{a+b+1-l\choose a}(mi)^l[(-1)^{a+b+l}+1]\zeta_l  \\
	&= \sum_{a,b=0}^\infty \mu_a\mu_b \frac{(-1)^{b+1}+(-1)^{a+1}}{(mi)^{a+b+2}}\sum_{l=1}^{a+1}{a+b+1-l\choose b}(mi)^l[(-1)^{a+b+l}+1]\zeta_l
\end{align*}
by exchanging $a$ and $b$ in the second term. We want to extract the $n$th coefficient of the above, so we write $\text{(b)} = \sum_{n=1}^\infty \alpha_n \zeta_n$, and seek $\alpha_n$; thus for all $a\geq n-1$, only the $n$th term in $\sum_{l=1}^{a+1}{a+b+1-l\choose b}(mi)^l[(-1)^{a+b+l}+1]\zeta_l$ contributes, i.e.
\begin{align*}
	\alpha_n &= \sum_{b=0}^\infty\sum_{a=n-1}^\infty \mu_a\mu_b \frac{(-1)^{b+1}+(-1)^{a+1}}{(mi)^{a+b+2}}{a+b+1-n\choose b}(mi)^n[(-1)^{a+b+n}+1].
\end{align*}
We also re-index using $2c=a+b+n$, since only terms with even $a+b+n$ contribute. Furthermore, $c\geq n$ and $b \leq 2c-2n+1$, thus
\begin{align*}
	\alpha_n &= 2 \sum_{c=n}^\infty\sum_{b=0}^{2c-2n+1}\mu_{2c-b-n}\mu_b\frac{(-1)^{b+1}+(-1)^{b-n+1}}{(mi)^{2c-2n+2}}{2c-2n+1\choose b} \\
	&= 2\sum_{c=0}^\infty \sum_{b=0}^{2c+1} \mu_{2c+n-b}\mu_b\frac{(-1)^{b+1}+(-1)^{b-n+1}}{(mi)^{2c+2}}{2c+1\choose b} \\
	&= 2(1+(-1)^n)\sum_{c=0}^\infty \sum_{b=0}^{2c+1} \mu_{2c+n-b}\mu_b\frac{(-1)^{b+c}}{m^{2c+2}}{2c+1\choose b}.
\end{align*} 
Finally, we write
\begin{align*}
	H &= \sum_{n=1}^\infty\bigg\{ \sum_{b=0}^{n-2} \mu_b\mu_{n-b-2}  - \frac{\beta_2}{2}(1+(-1)^n)\sum_{c=0}^\infty\sum_{b=0}^{2c+1}{2c+1\choose b}\frac{(-1)^{b+c}}{m^{2c+2}} \mu_{2c+n-b}\mu_b  \\&\qquad-\frac{4}{\beta}\sum_{j=0}^{\deg V'}\mu_{n+j-1}g_{j+1}\bigg\}\zeta_n.
\end{align*}
Such an equation is only possible if each coefficient of $\zeta_n$ are zero, including the constant coefficient $H$ itself, hence we arrive at the SDE. 

\begin{theorem}
	For all $n\geq 1$ we have
	\begin{align*}
		0 &= \sum_{b=0}^{n-2} \mu_b\mu_{n-b-2}  - \frac{\beta_2}{2}(1+(-1)^n)\sum_{c=0}^\infty\sum_{b=0}^{2c+1}{2c+1\choose b}\frac{(-1)^{b+c}}{m^{2c+2}} \mu_{2c+n-b}\mu_b  \\&\qquad-\frac{4}{\beta}\sum_{j=0}^{\deg V'}\mu_{n+j-1}g_{j+1}.
	\end{align*}
\end{theorem}

So far, we have not used any properties of the potential $S$, besides that the choice of $S$ must have an equilibrium measure with compact support. Let us assume that $S$ is even and therefore $V$ is even, then
\begin{align}
	0 &= \sum_{b=0}^{2n-2} \mu_b\mu_{2n-b-2}  - \frac{\beta_2}{2}(1+(-1)^{2n})\sum_{c=0}^\infty\sum_{b=0}^{2c+1}{2c+1\choose b}\frac{(-1)^{b+c}}{m^{2c+2}} \mu_{2c+2n-b}\mu_b \nonumber\\&\qquad -\frac{4}{\beta}\sum_{j=0}^{\deg V'}\mu_{2n+j-1}g_{j+1}\nonumber \\
	&= \sum_{b=0}^{n-1} \mu_{2b}\mu_{2n-2b-2}  - \beta_2\sum_{c=0}^\infty\sum_{b=0}^{c}{2c+1\choose 2b}\frac{1}{(-m^2)^{c+1}} \mu_{2c+2n-2b}\mu_{2b} \nonumber\\&\qquad -\frac{4}{\beta}\sum_{j=1}^{\deg V/2}\mu_{2n+2j-2}g_{2j} \label{fermionloopeqs}
\end{align}
for all integers $n\geq 1$. 

\subsection{Perturbative Analysis}

We now wish to examine the SDE perturbatively. What we find are converge series solutions in the examples studied. Note that we are interested in potentials whose coefficients may depend on the moments up to some finite order. For the quartic Dirac ensemble under consideration, the derivative of the potential which appears is $8g_4x^3 + (24g_4 \mu_2 +4g_2)x$. For this reason, we consider potentials of the form 
\[ V(x)=\sum_{j=1}^{d} \frac{\widetilde{t}_{2j}}{2j}x^{2j} \]
where 
\[ \widetilde{t}_{2j} = \bigg( t_{0,2j} + t_{2,2j}\mu_2 + \cdots t_{2d,2j}\mu_{2d}\bigg)  = \sum_{i=0}^{d} t_{2i,2j}\mu_{2i}. \]
We recover the quartic $(0,1)$ Dirac ensemble when $t_{0,2}=4g_2$, $t_{2,2}=24g_4$, $t_{0,4}=8g_4$, while all other $t_{2i,2j}=0$. We will expand in $t_2 \coloneqq t_{0,2}$. The SDE, for $n\geq 1$, is given by
\begin{align*}
	\frac{4}{\beta}\bigg(&t_2\mu_{2n}+\sum_{\substack{i=0\\ j=1}}^{d}{\vphantom{\sum}}'t_{2i,2j} \mu_{2n+2j-2}\mu_{2i}\bigg)\\&\qquad=  \sum_{b=0}^{n-1} \mu_{2b}\mu_{2n-2b-2}  - \beta_2\sum_{c=0}^\infty\sum_{b=0}^{c}{2c+1\choose 2b}\frac{1}{(-m^2)^{c+1}} \mu_{2c+2n-2b}\mu_{2b}. 
\end{align*}
The primed sum $\sum'$ indicates that we exclude the term where $i=0,j=1$ from the sum. We now define rescaled moments $M_j \coloneqq t_2^j \mu_{2j}$, so that the SDE is written as
\begin{align*}
	\frac{4}{\beta}\bigg(&M_n+\sum_{\substack{i=0\\ j=1}}^{d}{\vphantom{\sum}}' \frac{t_{2i,2j}}{t_2^{i+j}} M_{n+j-1}M_{i}\bigg)\\&\qquad=  \sum_{b=0}^{n-1}M_{b}M_{n-b-1}  - \beta_2\frac{1}{t_2}\sum_{c=0}^\infty\frac{1}{(-m^2)^{c+1}t_2^c}\sum_{b=0}^{c}{2c+1\choose 2b}  M_{c+n-b}M_{b} .
\end{align*}
We seek a solution of the form
\[ M_n = \sum_{k=0}^\infty \frac{1}{t_2^k}M_n^k \]
and study the SDE order by order in $t_2^{-1}$. Note that $M_0 = \mu_0 = \int_\Sigma \rho(x)\dif x = 1$, i.e. $M_0^k = 0$ for all $k\geq 1$. To this end, we introduce the formal generating function
\[ G(x) \coloneqq \sum_{n=0}^\infty M_n x^n = \sum_{k=0}^\infty \frac{1}{t_2^k} G^k(x) \]
where $G^k(x) = \sum_{n=0}^\infty M_n^k x^n$. We now multiply the $n$th SDE by $x^n$, and sum over $n=1,2,\dots$ to get the SDE in terms of the generating function $G(x)$. The left hand side becomes
\begin{align*}
	LHS &= \frac{4}{\beta}\bigg(\sum_{n=1}^\infty x^nM_n + \sum_{n=1}^\infty x^n\sum_{\substack{i=0\\ j=1}}^{d}{\vphantom{\sum}}' \frac{t_{2i,2j}}{t_2^{i+j}}M_{n+j-1}M_i \bigg)\\ &= \frac{4}{\beta}\bigg( G(x)-1 + \sum_{\substack{i=0\\ j=1}}^{d}{\vphantom{\sum}}' \frac{t_{2i,2j}}{t_2^{i+j}}\sum_{n=j}^\infty x^{n+1-j}M_n M_i\bigg) \\
	&= \frac{4}{\beta}\bigg( G(x)-1 + \sum_{\substack{i=0\\ j=1}}^{d}{\vphantom{\sum}}' \frac{t_{2i,2j}}{t_2^{i+j}}x^{1-j}M_i \bigg[ G(x) - \sum_{n=0}^{j-1} x^nM_n \bigg]\bigg)
\end{align*}
and the right hand side is now
\begin{align*}
	RHS &=  \sum_{n=1}^\infty x^n \sum_{b=0}^{n-1}M_{b}M_{n-b-1}  - \beta_2\frac{1}{t_2}\sum_{n=1}^\infty x^n\sum_{c=0}^\infty\frac{1}{(-m^2)^{c+1}t_2^c}\sum_{b=0}^{c}{2c+1\choose 2b}  M_{c+n-b}M_{b}  \\
	&= x\sum_{n=0}^\infty  \sum_{b=0}^n x^nM_bM_{n-b} - \beta_2\frac{1}{t_2}\sum_{c=0}^\infty \frac{1}{(-m^2)^{c+1}t_2^c}\sum_{b=0}^c {2c+1\choose 2b} M_b \sum_{n=1+c-b}^\infty x^{n+b-c}M_n \\
	&= xG(x)^2 - \beta_2\frac{1}{t_2}\sum_{c=0}^\infty \frac{1}{(-m^2)^{c+1}t_2^c}\sum_{b=0}^c {2c+1\choose 2b} x^{b-c} M_b \bigg[ G(x)-\sum_{n=0}^{c-b}x^n M_n \bigg].
\end{align*}
We now solve for the first few orders of $t_2^{-1}$. At $O(1)$ we have:
\begin{align}
	\frac{4}{\beta}(G^0(x)-1)&=xG^0(x)^2 \nonumber\\ G^0(x) &= \frac{2}{\beta x} \bigg(1 - \sqrt{1-\beta x}\bigg). \label{o1eqn}
\end{align}
Note that we took the $-$ sign so that $M_0=1$. At $O(t_2^{-1})$ we have:
\begin{align}
	\frac{4}{\beta} G^1(x) &= 2xG^0(x)G^1(x) + \frac{\beta_2}{m^2}(G^0(x)-1) \nonumber\\
	G^1(x) &= \frac{\frac{\beta_2}{ m^2}(G^0(x)-1)}{\frac{4}{\beta}-2xG^0(x)}.\label{o2eqn}
\end{align}
At $O(t_2^{-2})$ we find:
\begin{align*}
	G^2(x) &= \frac{1}{\frac{4}{\beta}-2xG^0}\cdot \bigg( \frac{\beta_2G^1}{m^2} - \frac{\beta_2}{m^4}\bigg[ \frac{1}{x}(G^0-1-xM_1^0) + 3M_1^0(G^0-1)\bigg]  \\ &+ x(G^1)^2 -\frac{4}{\beta}\bigg[ t_{2,2}M_1^0(G^0-1) +\frac{t_{0,4}}{x}(G^0-1-xM_1^0) \bigg]\bigg) 
\end{align*}
Higher order calculations are similar, but get more complex rapidly. \\

Let us now consider the simple Gaussian case, i.e. $t_{2i,2j}=0$ for all $i,j$, except $t_{0,2}\coloneqq t_2$, in which case $V'(x)=t_2 x$. The rescaled SDE for the generating function $G$ is
\begin{align*}
	\frac{4}{\beta}(G(x)-1) = xG(x)^2 - \beta_2\frac{1}{t_2}\sum_{c=0}^\infty \frac{1}{(-m^2)^{c+1}t_2^c}\sum_{b=0}^c {2c+1\choose 2b} x^{b-c} M_b \bigg[ G(x)-\sum_{n=0}^{c-b}x^n M_n \bigg].
\end{align*}
At $O(1)$ and $O(t_2^{-1})$ we again arrive at \ref{o1eqn} and \ref{o2eqn} respectively. At $O(t_2^{-2})$, they become
\begin{align*}
	G^2(x) &= \frac{x(G^1)^2 + \bigg(\frac{\beta_2G^1}{m^2} - \frac{1}{m^4}(x^{-1}[G^0-1-xM_1^0] + 3M_1^0[G^0-1]) \bigg)}{\frac{4}{\beta}-2xG^0}.
\end{align*}
At $O(t_2^{-3})$, we again proceed along the same lines. The contributing terms for each $c=0,1,2$ are:
\begin{align*}
	c = 0 :\quad& \frac{1}{(-m^2)}[G^2(x)]\\ 
	c = 1 :\quad& \frac{1}{(-m^2)^2} \bigg( x^{-1}[G^1(x)-xM_1^1] + 3(M_1^1[G^0(x)-1]+M_1^0[G^1(x)])\bigg)\\
	c = 2 :\quad& \frac{1}{(-m^2)^3} \bigg( x^{-2}[G^0(x)-1-xM_1^0-x^2M_2^0]\\ &\qquad\qquad+10x^{-1}M_1^0[G^0(x)-1-xM_1^0]+5M_2^0[G^0(x)-1] \bigg),
\end{align*}
we label them $c_0,c_1,c_2$ respectively. We then find that
\begin{align*}
	\frac{4}{\beta}G^3(x) &= 2x(G^0G^3 + G^1G^2) - \beta_2(c_0 + c_1 + c_2) \\
	G^3(x) &= \frac{2xG^1G^2 - \beta_2(c_0+c_1+c_2)}{\frac{4}{\beta}-2xG^0}.
\end{align*}

This procedure can continue to be carried out to obtain successively higher order corrections to the moments. 

\section{The Gaussian model}\label{sec: Gaussian}

In the case of the Gaussian potential $V(z)=\frac{g_2}{2}z^2$ with $\Sigma = [-L,L] =  \text{Supp} \rho$, there is a more direct method to investigate the saddle point equation \eqref{saddle} when $\beta_2 = \beta=2$, in which case \ref{saddle} reads (with $m=1$):
\begin{align}
	2g_2 z = W(z+i0) + W(z-i0) + [W(z+i) + W(z-i)] \label{fermionic saddle}
\end{align}
for the fermionic matrix integral, and for the bosonic case reads
\begin{align}
	2g_2z = W(z+i0) + W(z-i0) - [W(z+i) + W(z-i)]. \label{bosonic saddle} 
\end{align}
The bosonic case is exactly what is studied in \cite{kostov2002exact} and \cite{kazakov1999d}, and even earlier in \cite{hoppe1989quantum}, and is sometimes referred to as the Hoppe model for that reason. The idea is to use the saddle point equation to create an analytic function which will turn out to be the inverse of a Schwarz-Christoffel transformation, allowing one to express quantities of interest, such as the second moment $\mu_2$, in terms of elliptic integrals. We will proceed along the same lines as \cite{kazakov1999d} while spelling out mathematically rigorous details. 

Note that we are able to write the second moment of this model explicitly in terms of elliptic functions in terms of $g_2$. From the second moment we are able to obtain the free energy since in the limit by integrating both sides of the identity
\begin{equation*}
    \frac{\partial}{\partial g_{2}}\ln Z = \mu_{2}(g_{2}).
\end{equation*}
\subsection{Bosonic Function}

We begin with the bosonic case, as it is simpler and already appears in the literature. Recall that the saddle point equation in this case is 
\begin{align*}
	2g_2 z = W(z+i0) + W(z-i0) - [W(z+i) + W(z-i)].
\end{align*}
We now define a function
\begin{align*}
	B(z) \coloneqq g_2 z^2 + i \bigg[W(z+i/2)-W(z-i/2)\bigg]
\end{align*}
which results in a saddle point equation for $B$: for $a\in\Sigma$
\begin{align*}
	B(a+i/2 + i0) = B(a-i/2-i0).
\end{align*}
Indeed,
\begin{align*}
	B(a+ i/2+i0) &=  g_2(a^2+ia-1/4) +i\bigg[W(a+i)-W(a+i0)\bigg] \\
	B(a-i/2-i0)  &= g_2(a^2 - ia - 1/4) +i\bigg[W(a-i0)-W(a-i)\bigg] 
\end{align*}
and so 
\begin{align*}
	B(a+i/2&+i0)-B(a-i/2-i0)\\ &= \underbrace{2ig_2a + i\bigg[-W(a+i0)-W(a-i0)+W(a+i)+W(a-i)\bigg]}_{\equiv 0\text{ for $a\in\Sigma$}}.
\end{align*}
Since $\rho$ is even, it follows that $W(z)$ is odd $W(-z)=-W(z)$, and we also have $W(\overline{z})=\overline{W(z)}$, since
\begin{align*}
	W(\overline{z})=\int_\Sigma \frac{\rho(y)\dif y}{\overline{z}-y} = \overline{\int_\Sigma \frac{\rho(y)\dif y}{z-y}}=\overline{W(z)}.
\end{align*}
It follows that $B(\overline{z})=\overline{B}(z)$, and $B(-z)=B(z)$, since
\begin{align*}
	B(\overline{z}) &= g_2\overline{z}^2 + i\bigg[W(\overline{z}+i/2)-W(\overline{z}-i/2)\bigg]= g_2\overline{z^2} + i\bigg[\overline{W(z-i/2)}-\overline{W(z+i/2)}\bigg] = \overline{B(z)} \\
	B(-z) &= g_2z^2 +i\bigg[W(-z+i/2)-W(-z-i/2)\bigg] = g_2z^2 +i\bigg[-W(z-i/2)+W(z+i/2)\bigg]=B(z)
\end{align*}
These properties imply that $B(z)$ is real for $z\in \rr$, $z\in i\rr$ and $z=\pm\frac{i}{2}+a\pm i0$ for $a\in \Sigma$. Indeed, if $x\in \rr$, then
\begin{align*}
	\overline{B(x)} = B(\overline{x})=B(x)
\end{align*}
similarly, if $z=iy$ for $y \in \rr$, then
\begin{align*}
	\overline{B(iy)}=B(\overline{iy})=B(-iy)=B(iy)
\end{align*}
and if $z=\pm\frac{i}{2}+a\pm i0$ for $a\in \Sigma$ then, by the saddle point equation \ref{bosonic saddle} we have
\begin{align*}
	\overline{B(a\pm i/2 \pm i0)} = B(a \mp i/2 \mp i0) = B(a \pm i/2 \pm i0).
\end{align*}

This leads us to consider the polygonal domain $U\subset \cc$, bounded by $\rr_+$, $i\rr_+$ and the segment $\frac{i}{2}+a$ for positive $a\in \Sigma$. In fact, it turns out that the behavior of $B$ is enough to conclude that it is a biholomorphism $B:U\to \hh$.

\begin{figure}\label{bosonicdomain}
\begin{center}
	\begin{tikzpicture}
		\draw[->] (0,0)--(5,0);
		\draw[->] (0,0)--(0,5.1);
		\draw (0,2)--(2,2);
		\draw (-.5,0) node {$0$};
		\draw (3.5,3.5) node {\huge{$U$}};
		\draw (6,0) node {$\re z$};
		\draw (0,5.5) node {$\im z$};
		\draw (2.5,2) node {$\frac{i}{2}+a$};
		\draw (-.5,2) node {$\frac{i}{2}$};
	\end{tikzpicture}
\end{center}
\caption{Domain for the bosonic function $B$}
\end{figure}
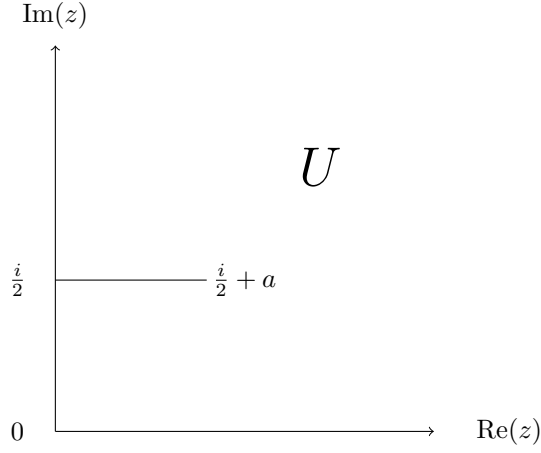

\subsection{Fermionic Function}

The fermionic saddle point equation reads
\begin{align*}
	2g_2z = W(z+i0) + W(z-i0) + [W(z+i) + W(z-i)]. 
\end{align*} 
We define the function
\begin{align*}
	F(z) \coloneqq g_2z + i[ W(iz + i/2) + W(iz - i/2) 
\end{align*} 
which gives us a saddle point equation for $F$: for $a\in \Sigma$ we have
\begin{align*}
	F(ia+1/2\pm 0) = -F(ia - 1/2 \mp 0) 
\end{align*} 
Again, we verify directly,
\begin{align*}
	F(ia + 1/2 \pm 0 ) &= ig_2a + \frac{1}{2} + i[W(-a+i)+W(-a\pm i0)] = ia + \frac{1}{2} -i[W(a-i)+W(a\mp i0)]\\
	F(ia - 1/2 \pm 0)  &= ig_2a - \frac{1}{2} + i[W(-a\pm i0)+W(-a-i)]= ia - \frac{1}{2} - i[W(a\mp i0)+W(a+i)],
\end{align*}
thus
\begin{align*}
	F(ia+1/2&\pm 0) + F(ia - 1/2 \mp 0) \\&
	= i\bigg\{	2g_2a - [W(a-i)+W(a\mp i0) + W(a\pm i0) + W(a+i)]  \bigg\} = 0 	
\end{align*}
for any $a\in \Sigma$, by the saddle point equation.

We likewise show that $F(x)$ is real for $x\in \rr$,
\begin{align*}
	\overline{F(x)} &= g_2 x - i [ \overline{W(ix+i/2)+W(ix-i/2)} ]
	= g_2x + i[W(ix+i/2)+W(ix-i/2)] = F(x).
\end{align*}
Also, by observing that for $a\in \Sigma$
\begin{align*}
	\overline{F(ia+1/2\pm 0)} &= -ig_2a + \frac{1}{2} + i[W(a+i)+W(a\pm i0)] \\
	\overline{F(ia-1/2\pm 0)} &= -ig_2a - \frac{1}{2} +i[W(a\pm i0)+W(a-i)]
\end{align*}
we see that
\begin{align*}
	F(ia&+1/2\pm 0)-\overline{F(ia+1/2\pm 0)} \\
	&= i(2g_2a - i [W(a-i)+W(a\mp i0)+W(a+i)+W(a\pm i0)]) = 0 \\
	F(ia&-1/2\pm 0)-\overline{F(ia+1/2\pm 0)} \\
	&= i(2g_2a - i [W(a\pm i0)+W(a+i)+W(a\pm i0)+W(a-i)]) = 0
\end{align*}
showing that $F$ is real along the cuts $\pm1/2 + i\Sigma$. Again, we are led to a polygonal domain $V$ which is bounded by the real axis, as well as the segments $1/2 +ai$ and $-1/2+ai$ for $a$ positive and in $\Sigma$. Again, we will be able to conclude that $F:V\to\hh$ is a biholomorphism. 

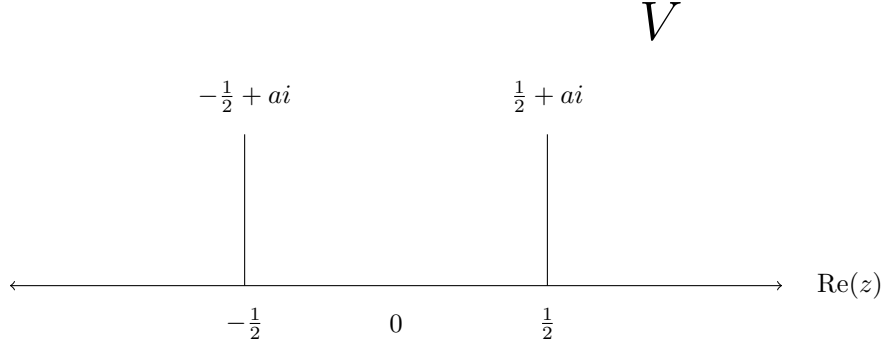
\begin{figure}[h]\label{fermionicdomain}
\begin{center}
	\begin{tikzpicture}
		\draw[<->] (-5.1,0)--(5.1,0);
		\draw (-2,0)--(-2,2);
		\draw (-2,2.5) node {$-\frac{1}{2}+ai$};
		\draw (2,0)--(2,2);
		\draw (2,2.5) node {$\frac{1}{2}+ai$};
		\draw (-2,-.5) node {$-\frac{1}{2}$};
		\draw (2,-.5) node {$\frac{1}{2}$};
		\draw (0,-.5) node {$0$};
		\draw (3.5,3.5) node {\huge{$V$}};
		\draw (6,0) node {$\re z$};
	\end{tikzpicture}
\end{center}
\caption{Domain for the fermionic function $F$}
\end{figure}

\subsection{$B$ and $F$ as inverse Schwarz-Christoffel Maps}

We are now ready to show that the functions $B$ and $F$ constructed above are indeed biholomorphisms which are given by the inverse of a Schwarz-Christoffel integral. 

\begin{prop}
	The functions $B$ and $F$ are biholomorphisms from the domains $U$ and $V$ respectively to $\hh$. 
\end{prop}

\begin{proof}
	We only need to show that $B$ and $F$ are bijective, since any bijective holomorphic map is a biholomorphism. We make the argument for $B$ only, as the argument with $F$ goes through with superficial modifications. \\
	
	We have already seen that on the boundary of $U$, $B(z)$ is real, hence $B(\partial U)\subset \rr$. Moreover, since $B$ is holomorphic on $U$, the real and imaginary parts are harmonic functions on $U$. As $|z|\to\infty$, we know $W(z)\to 0$, and therefore the $g_2z^2=g_2(x^2-y^2)+ig_2(2xy)$ term dominates for $z\in U$ large, implying that $v\coloneqq\im{B}>0$ for at least some points in $U$. Consider the domains $U_R$ given by $U\cap B(0;R)$, and $R$ large enough to avoid the cut at $i/2+\Sigma$. Let $A_R$ denote the boundary arc $U\cap\overline{B(0,R)}$. By the maximum principle for harmonic functions, the minimum/maximum of $v$ must be achieved on $\partial U_R$. For any $\varepsilon>0$, choose $R$ such that $|\text{Im}( i[W(z+i/2)-W(z-i/2)])|<\varepsilon$ on $\partial B(0,R)$, so that $v(z) > 2g_2xy - \varepsilon > - \varepsilon$ on the boundary arc $A_R$. Since $v\equiv 0$ along the rest of the boundary (real and imaginary axes, as well as the cut $i/2 + \Sigma$), we find that on $U$, we must have $v(z)>0$, hence $B$ maps $U$ into $\hh$.  \\
	
	Next, we note that for $x,y\in (0,\infty)$, as $x,y\to \infty$ we have $B(x) \to \infty$ and $B(iy) \to -\infty$. Since $\partial U$ is connected, and $B$ is continuous along the boundary, $B(\partial U)\subset \rr$ is also connected which means that $B(\partial U) = \rr$. Let $\dd\subset \cc$ be the unit disk. We now view all these domains $U,\hh,\dd$ as subsets of the Riemann sphere, $\hat{\cc}$. The Riemann mapping theorem gives us bioholomorphisms $\phi: U \to \dd$ and $\psi: \hh\to \dd$. Furthermore, by the Caratheodory-Tohorst theorem (see, e.g. \cite{milnor1990dynamicscomplexvariable} Sections 15 and 16), these biholomorphisms extend to continuous functions along the boundaries, i.e. $\phi:\overline{U}\to\overline{\dd}$ and $\psi: \overline{\hh}\to\overline{\dd}$ are continuous. Define now $\alpha:\dd\to \dd$ through $\alpha= \psi\circ B\circ \phi^{-1}$. We find that $\alpha(\partial \dd) = \partial \dd$, and this curve has a winding number of $1$ around any point inside the disk. Thus, for any $z_0 \in D$, choose $R>|z_0|$ and through the homotopy invariance of winding numbers, we have
	\[ 1 = \frac{1}{2\pi i}\oint_{\alpha(\partial \dd)} \frac{1}{z-z_0}\dif z = \frac{1}{2\pi i}  \int_{\alpha(S^1_R)} \frac{\dif z}{z-z_0} \]
	where $S^1_R$ is a circle of radius $R$ enclosing $z_0$. Therefore
	\[ 1 = \frac{1}{2\pi i} \oint_{\alpha(S^1_R)}\frac{\dif z}{z-z_0} = \frac{1}{2\pi i}\oint_{S^1_R}\frac{(\alpha(z)-z_0)'}{\alpha(z)-z_0}\]
	which then the argument principle shows that $\alpha(z)-z_0$ has exactly one zero on $\dd$ (since it has no poles by assumption), hence is bijective, and therefore so is $B(z)$, proving that $B(z)$ is a biholomorphism. 
\end{proof}

The Schwarz-Christoffel integral defines a biholomorphism $f_{SC}:\hh\to D$ where $\hh=\{z\in\cc:\im{z}>0\}$ is the upper half plane, and $D$ is any polygonal domain through the following procedure:
\begin{itemize}
	\item There are points $x_1>\dots>x_n\in \rr\cup\{\infty\}$ which are mapped to the vertices $\{\zeta_i\}_{i=1}^n$ of the polygonal domain, $f_{SC}(x_i)=\zeta_i$, for $i=1,2,\dots,n$ (points at infinity are allowed). The vertices are to be ordered by starting at $\zeta_1$ and following the polygon counterclockwise.
	\item To each vertex $\zeta_i$, an angle $\pi\alpha_i$ is associated. If $\zeta_i\neq \infty$ then the angle is simply the angle at the vertex, and if $\zeta_i = \infty$ the angle is given by minus the angle between the two vectors which lie along the segments connecting the point at $\infty$. 
	\item The map $f_{SC}$ is given by the Schwarz-Christoffel integral, assuming all $x_i\neq \infty$
	\begin{align*}
		f_{SC}(z) = A\int_{0}^z (t-x_1)^{\alpha_1-1}(t-x_2)^{\alpha_2-1}\cdots (t-x_n)^{\alpha_n-1} \dif t + C
	\end{align*}
	with $A,C$ constants. If $x_i=\infty$, the factor $(z-x_i)^{\alpha_i-1}$ is omitted from the above equation. 
\end{itemize}
The Riemann mapping theorem guarantees the existence of a biholomorphism between any two simply connected domains $D\neq \cc$ and $D'\neq \cc$. Moreover, such a biholomorphism $f:D\to D'$ can be uniquely determined by the values it takes for 3 points on the boundary of $D$. Therefore, if we take the Schwarz-Christoffel map $f_{SC}:\hh \to U$ such that $f_{SC}(x_2)=\frac{i}{2}, f_{SC}(x_3)=\frac{i}{2}+L$ and $f_{SC}(x_4)=\frac{i}{2}$, we have $f_{SC}=B^{-1}$ (and so $f_{SC}(x_1)=0$ and $f_{SC}(\infty)=\infty$). 

As $z$ has been used to denote points in $U$ or $V$, we will write $z=z(\zeta)=f_{SC}(\zeta)$. The associated angles are $\pi\alpha_1 = \pi\alpha_2 = \pi\alpha_4=\pi/2$ and $\pi\alpha_3=2\pi$ (the angle at $\infty$ is $-\pi\alpha_5=\pi/2$ but is not needed since the associated factor will be omitted). The same procedure applies to $F$, and we summarize the data in the tables below: 

\begin{table}[h]
\begin{center}
	\begin{tabular}{|l|l|l|}
		\hline
		$\zeta=B(z)$ & Vertex $z\in U$ & Angle $\pi\alpha_i$ \\
		\hline
		$\infty$ & $\infty$ & $\pi\cdot \frac{1}{2}$ \\
		$x_1$ & $0$ & $\pi\cdot \frac{1}{2}$ \\
		$x_2$ & $i/2$ & $\pi\cdot \frac{1}{2}$ \\
		$x_3$ & $i/2+L$ & $\pi \cdot 2$\\
		$x_4$ & $i/2$ & $\pi\cdot\frac{1}{2}$\\
		\hline
	\end{tabular}\quad		\begin{tabular}{|l|l|l|}
		\hline
		$\zeta=F(z)$ & Vertex $z\in V$ & Angle $\pi\alpha_i$ \\
		\hline
		$\infty$ & $\infty$ & $\pi\cdot 1$ \\
		$y_1$ & $1/2$ & $\pi\cdot \frac{1}{2}$ \\
		$y_2$ & $1/2+iL$ & $\pi\cdot 2$ \\
		$y_3$ & $1/2$ & $\pi \cdot \frac{1}{2}$\\
		$y_4$ & $-1/2$ & $\pi\cdot\frac{1}{2}$\\
		$y_5$ & $-1/2+iL$ & $\pi\cdot 2$\\
		$y_6$ & $-1/2$ & $\pi\cdot\frac{1}{2}$\\
		\hline
	\end{tabular}
\end{center}
\end{table}

Thus, $B(z)=\zeta$ may be inverted to obtain
\begin{align}
	z &= A\int_0^\zeta (t-x_1)^{-1/2}(t-x_2)^{-1/2}(t-x_3)(t-x_4)^{-1/2} \dif t + C \nonumber \\
	&= A\int_{x_1}^\zeta \frac{(t-x_3)\dif t}{\sqrt{(t-x_1)(t-x_2)(t-x_4)}} \label{scintegral}
\end{align}
where we used $z(x_1)=0$ to determine $C$. \\

We can also use the conditions $z(x_2) = z(x_4)=i/2$ and $z(x_3) = i/2 + L$ to derive additional constraints:
\begin{align}
	z(x_2) = \frac{i}{2} &= A\int_{x_1}^{x_2} \frac{\dif t (t-x_3)}{\sqrt{-(x_1-t)(t-x_2)(t-x_4)}}\nonumber \\
	 \frac{1}{2} &= A\int_{x_2}^{x_1} \frac{\dif t (t-x_3)}{\sqrt{(x_1-t)(t-x_2)(t-x_4)}} \label{c1} \\
	z(x_3) = \frac{i}{2}+L &= A\int_{x_1}^{x_3} \frac{\dif t (t-x_3)}{\sqrt{(x_1-t)(x_2-t)(t-x_4)}} \nonumber\\
	z(x_4) = \frac{i}{2} &= A\int_{x_1}^{x_4} \frac{\dif t (t-x_3)}{\sqrt{(x_1-t)(x_2-t)(t-x_4)}}. \nonumber
\end{align}
The last two constraints combine to give
\begin{align}
	z(x_3)-z(x_4) = L = A\int_{x_4}^{x_3} \frac{\dif t (t-x_3)}{\sqrt{(x_1-t)(x_2-t)(t-x_4)}} \label{c2}
\end{align}
while the first two give
\begin{align}
	z(x_2)-z(x_3)= -L  = A\int_{x_3}^{x_2} \frac{\dif t (t-x_3)}{\sqrt{(t-x_1)(t-x_2)(t-x_4)}}. \label{c3}
\end{align}

Finally, inverting $F(z) = \zeta$ yields
\begin{align}
	z &= A\int_0^\zeta (t-y_1)^{-\frac{1}{2}}(t-y_2)^1(t-y_3)^{-\frac{1}{2}}(t-y_y)^{-\frac{1}{2}}(t-y_5)^1 + C\nonumber\\
	&= A\int_0^\zeta \frac{(t-y_2)(t-y_5)\dif t}{\sqrt{(t-y_1)(t-y_3)(t-y_4)(t-y_6)}} + C.\label{fermionfn}
\end{align}
We also get constraints from the table, $y_i = F(z_i)$, and we can subtract $(i+1)$th from the $i$th constraint to get the following 5 equations:
\begin{align*}
	L &= A\int_{y_2}^{y_1}\frac{(t-y_2)(t-y_5)\dif t}{\sqrt{(y_1-t)(t-y_3)(t-y_4)(t-y_6)}} \\
	-L &= A\int_{y_3}^{y_2} \frac{(t-y_2)(t-y_5)\dif t}{\sqrt{(y_1-t)(t-y_3)(t-y_4)(t-y_6)}}\\
	1 &= A\int_{y_4}^{y_3} \frac{(t-y_2)(t-y_5)\dif t}{\sqrt{(y_1-t)(y_3-t)(t-y_4)(t-y_6)}}\\
	L &= A\int_{y_5}^{y_4} \frac{(t-y_2)(t-y_5)\dif t}{\sqrt{(y_1-t)(y_3-t)(y_4-t)(t-y_6)}}\\
	-L &= A\int_{y_6}^{y_5} \frac{(t-y_2)(t-y_5)\dif t}{\sqrt{(y_1-t)(y_3-t)(y_4-t)(t-y_6)}}.
\end{align*}
We can further sum the first two of the above equations to obtain
\begin{align}
	0 &= \int_{y_3}^{y_1} \frac{(t-y_2)(t-y_5)\dif t}{\sqrt{(y_1-t)(t-y_3)(t-y_4)(t-y_6)}} \label{fermionconstraint1}
\end{align}
as well as the last two,
\begin{align}
	0 &= \int_{y_6}^{y_5} \frac{(t-y_2)(t-y_5)\dif t}{\sqrt{(y_1-t)(y_3-t)(y_4-t)(t-y_6)}}. \label{fermionconstraint2}
\end{align}
In addition, we are left with the third equation still,
\begin{align}
	1 &= A\int_{y_4}^{y_3} \frac{(t-y_2)(t-y_5)\dif t}{\sqrt{(y_1-t)(y_3-t)(t-y_4)(t-y_6)}}. \label{fermionconstraint3}
\end{align}

A general elliptic integral is an integral of the form
\begin{align*}
	\int R(t,y)\dif t
\end{align*}
where $R$ is rational in $t$ and $y$, $y=y(t)$ is a function such that $y^2$ is cubic or quartic in $t$ with non-repeating zeroes. E.g. in $\ref{fermionconstraint3}$, 
\begin{align*}
	y = \sqrt{(y_1-t)(y_3-t)(t-y_4)(t-y_6)} \qquad R(t,y) = \frac{(t-y_2)(t-y_5)}{y}.
\end{align*}
All of the constraints expressed here are therefore general elliptic integrals, where $y(t)$ is cubic for the bosonic case and quartic for the fermionic case. Such elliptic integrals and their properties have been studied in detail, and we can in principle reduce these constraints down to equations involving elliptic integrals in their standard form, depending on some parameter(s).

\subsection{Asymptotics and Lagrange Inversion}\label{bosonanalysis}

Given the integral representation of $z(\zeta)$ we may deduce the asymptotics as $\zeta\to\infty$. Start by differentiating $z$,
\begin{align*}
	z'(\zeta) &= A(\zeta -x_3)\big[(\zeta-x_1)(\zeta-x_2)(\zeta-x_4)\big]^{-\frac{1}{2}} \nonumber\\
	&= \frac{A(\zeta-x_3)}{\zeta^{\frac{3}{2}}}\frac{1}{(1-x_1/\zeta)^{\frac{1}{2}}}\frac{1}{(1-x_2/\zeta)^{\frac{1}{2}}}\frac{1}{(1-x_4/\zeta)^{\frac{1}{2}}} \nonumber\\
	&\sim \frac{A(\zeta-x_3)}{\zeta^{\frac{3}{2}}}\bigg[ \sum_{k_1=0}^\infty {-\frac{1}{2} \choose k_1} (-x_1\zeta^{-1})^{k_1} \bigg]\bigg[ \sum_{k_2=0}^\infty {-\frac{1}{2} \choose k_2} (-x_2\zeta^{-1})^{k_2} \bigg]\nonumber\\&\qquad \cdot\bigg[ \sum_{k_4=0}^\infty {-\frac{1}{2} \choose k_4} (-x_4\zeta^{-1})^{k_4} \bigg]\nonumber \\
	&= A(\zeta^{-\frac{1}{2}}-x_3\zeta^{-\frac{3}{2}})\bigg[  \sum_{k=0}^\infty (-1)^k \zeta^{-k}\underbrace{\sum_{p+q+r = k} {-\frac{1}{2} \choose p}{-\frac{1}{2} \choose q}{-\frac{1}{2} \choose r} x_1^px_2^qx_4^r}_{\coloneqq \gamma_k}\bigg] \nonumber\\
	&= A\sum_{k=0}^\infty (-1)^k \gamma_k (\zeta^{-\frac{1}{2}-k}-x_3\zeta^{-\frac{3}{2}-k}) = A\bigg[ \zeta^{-\frac{1}{2}} + \sum_{k=1}^\infty (-1)^{k} (\gamma_{k}+x_3\gamma_{k-1})\zeta^{-\frac{1}{2}-k} \bigg],
\end{align*}
therefore, we have the asymptotics for $z$
\begin{align}
	z(\zeta)&\sim 2A \bigg[a_0 +  \zeta^{\frac{1}{2}} + \sum_{k=1}^\infty \underbrace{ \frac{(-1)^{k+1}(\gamma_k + x_3\gamma_{k-1}) }{2k-1}}_{\coloneqq a_k} \zeta^{\frac{1}{2}-k}\bigg]. \label{integral expansion}
\end{align}

On the other hand, we know the asymptotics of 
\[\zeta(z)=B(z) = g_2z^2 + i [W(z+\frac{i}{2})-W(z-\frac{i}{2})]\] 
through the expansion of $W$ at $\infty$, i.e. 
\[  W(z) = \sum_{k=0}^\infty \frac{\mu_{2k}}{z^{2k+1}}.  \]
Therefore, through Lagrange inversion, we will again arrive at an expansion for $z(\zeta)$, enabling us to relate the moments $\mu_{2k}$ with the parameters $x_i$ and their constraints. As our primary interest is in finding an expression for $\mu_2$, it suffices to consider the expansion to $\mcal{O}(z^{-4})$. We proceed along these lines, first by expanding $W(z\pm \frac{i}{2})$,
\begin{align*}
	W(z\pm \frac{i}{2}) &= \sum_{k=0}^\infty \frac{\mu_{2k}}{(z\pm \frac{i}{2})^{2k+1}} = \frac{1}{z}\frac{1}{1\pm \frac{i}{2z}}\sum_{k=0}^\infty \frac{\mu_{2k}}{z^{2k}(1\pm \frac{i}{2z})^{2k}}.
\end{align*}
We use the expansion
\[ \frac{1}{1\pm \frac{i}{2z}} \sim  1 - \bigg( \pm \frac{i}{2z} \bigg)  + \bigg( \pm\frac{i}{2z}  \bigg)^2 - \bigg( \pm\frac{i}{2z} \bigg)^3 + \mcal{O}(z^{-4}) \]
to find
\begin{align}
	W(z\pm \frac{i}{2}) &\sim \frac{1}{z}\bigg[1 -  \bigg(\pm \frac{i}{2z}\bigg) + \cdots \bigg]\sum_{k=0}^\infty \frac{\mu_{2k}}{z^{2k}}\bigg[ 1 - \bigg( \pm\frac{i}{2z} \bigg) + \cdots \bigg]^{2k} \nonumber\\
	&= \bigg[\frac{1}{z}- \bigg(\pm \frac{i}{2z^2}\bigg) - \frac{1}{4z^3} + \bigg( \pm\frac{i}{8z^4}\bigg) \bigg]\bigg[ 1 + \frac{\mu_2}{z^2} \mp \frac{i\mu_2}{z^3} \bigg]+ \mcal{O}(z^{-5})\nonumber\\
	&= \frac{1}{z} \mp \bigg( \frac{i}{2}\bigg)\frac{1}{z^2} + \bigg(\mu_2- \frac{1}{4}\bigg)\frac{1}{z^3} \pm\frac{i}{2} \bigg( \frac{1}{4} - 3\mu_2  \bigg)\frac{1}{z^4} +\mcal{O}(z^{-5})\label{Wexpansion}
\end{align}
hence, setting $\delta \coloneqq 3\mu_2 - \frac{1}{4}$ as in \cite{kazakov1999d}, we have
\begin{align*}
	B(z)=g_2z^2 + i\bigg[W(z+\frac{i}{2})-W(z-\frac{i}{2})\bigg] \sim g_2z^2 + \frac{1}{z^2} + \frac{\delta}{z^4} + \mcal{O}(z^{-6}).
\end{align*}

We now recall the Lagrange inversion theorem. Given a holomorphic function $f$ with series expansion of the form
\[ f(z) = \sum_{j=1}^\infty f_j z^j \]
such that $f_1\neq 0$, then there exists a (local) inverse $g$, i.e. $f(g(z))=z$, and the series expansion coefficients are related through
\[ k[z^k]g^n = n[z^{-n}]f^{-k} \]
where $[z^k]f$ denotes the $k$th coefficient of the series expansion, $f_k = [z^k]f$. \\

To apply this theorem to $B(z)$, set $w=z^2$ and consider the function $1/B(w)$:
\begin{align*}
	\frac{1}{B(w)} &\sim \frac{1}{g_2 w} \bigg(  \frac{1}{1 + \frac{1}{g_2}[ w^{-2} + \delta w^{-3} + \mcal{O}(w^{-4}) ]} \bigg) \\
	&\sim \frac{1}{g_2 w} - \frac{1}{g_2^2 w^3} + \frac{\delta}{g_2^2 w^4} + \mcal{O}(w^{-5})
\end{align*}
and now $f(w) \coloneqq \frac{1}{B(1/w)}$ is of the necessary form to apply Lagrange inversion,
\[ f(w) = \frac{w}{g_2} - \frac{w^3}{g_2^2} + \frac{\delta w^4}{g_2^2} + \mcal{O}(w^{5}). \] 
Let $g$ be the inverse of $f$, which has coefficients
\[ [w^k]g = \frac{1}{k}[w^{-1}]f^{-k} = \frac{1}{k}\res_{w=0} B(1/w)^k \]
which we may easily compute for $k=1,2,3,4$,
\begin{align*}
	[w]g &= \res_{w=0}\bigg[ \frac{g_2}{w} + w + \delta w^2 + \cdots \bigg] = g_2 \\
	[w^2]g &= \frac{1}{2}\res_{w=0} \bigg[ \bigg( \frac{g_2}{w} + w + \delta w^2 +\cdots\bigg)^2 \bigg]  = 0 \\
	[w^3]g &= \frac{1}{3}\res_{w=0}\bigg[ \bigg( \frac{g_2}{w} + w + \delta w^2 + \cdots \bigg)^3 \bigg] = g_2^2 \\
	[w^4]g &= \frac{1}{4}\res_{w=0}\bigg[\bigg( \frac{g_2}{w} + w + \delta w^2 + \cdots \bigg)^4\bigg] = \delta g_2^3.
\end{align*}
Thus 
\[ g(w) \sim g_2 w + g_2^2 w^3 + g_2^3\delta w^4 + \mcal{O}(w^5) \]
and since $z^2 = w = g(f(w)) = g(1/B(1/w))$, we have 
\[\frac{1}{z^2} = \frac{1}{w} = g(f(1/w)) = g(1/B(w))=g(1/\zeta)\]
and so
\begin{align}
	z^2 &\sim \frac{1}{\frac{g_2}{\zeta} + \frac{g_2^2}{\zeta^3} + \frac{g_2^3\delta}{\zeta^4} + \mcal{O}(\zeta^{-5})}  \nonumber\\
	&\sim \frac{\zeta}{g_2} \bigg( 1 + \frac{g_2}{\zeta^2} + \frac{g_2^2\delta}{\zeta^3} + \mcal{O}(\zeta^{-4}) \bigg)^{-1} \nonumber\\
	z &\sim \frac{\zeta^{\frac{1}{2}}}{\sqrt{g_2}}\bigg(  1 + \frac{g_2}{\zeta^2} + \frac{g_2^2\delta}{\zeta^3} + \mcal{O}(\zeta^{-4}) \bigg)^{-\frac{1}{2}} \nonumber\\
	&\sim \frac{1}{\sqrt{g_2}}\zeta^{\frac{1}{2}} - \frac{\sqrt{g_2}}{2}\zeta^{-\frac{3}{2}} - \frac{\sqrt{g_2^3}\delta}{2}\zeta^{-\frac{5}{2}} + \mcal{O}(\zeta^{-\frac{7}{2}}). \label{bosoninversion}
\end{align}
Note that we took the positive root since as $z\to \infty$, $\zeta\to\infty$ as well. We may now compare the above Lagrange inversion \ref{bosoninversion} and the asymptotic expansion \ref{integral expansion} to obtain the following formulae:
\begin{align}
	\mcal{O}(\zeta^{\frac{1}{2}}):\qquad&  \boxed{A = \frac{1}{2\sqrt{g_2}}} \nonumber\\
	\mcal{O}(\zeta^{0}):\qquad& \boxed{a_0 = 0} \nonumber\\
	\mcal{O}(\zeta^{-\frac{1}{2}}):\qquad&\boxed{2x_3  = x_1+x_2+x_4} \label{order-1/2} \\
	\mcal{O}(\zeta^{-\frac{3}{2}}):\qquad& \boxed{6g_2 =  x_1^2+x_2^2+x_4^2 - 2x_3^2} \label{order-3/2} \\
	\mcal{O}(\zeta^{-\frac{5}{2}}):\qquad &\boxed{\delta = -2\frac{\gamma_3 + x_3\gamma_2}{5g_2^2}} \label{order-5/2}
\end{align}
We postpone the full simplification of \ref{order-5/2} for the time being. 
In the above calculations we have used the generalization of ${r\choose k}$, 
\[ {r \choose k} = \frac{(r)_k}{k!} = \frac{1}{k!}\frac{\Gamma(r+1)}{\Gamma(r-k+1)} = \frac{1}{k!}r(r-1)\cdots(r-k+1). \]
Furthermore, the condition that $a_0 = 0$ can be directly verified by studying $z(\zeta)$ as $\zeta\to\infty$; from \ref{integral expansion} we see that 
\begin{align*}
	\lim_{\zeta\to\infty} \bigg( z(\zeta) - 2A\zeta^{\frac{1}{2}} \bigg) = 2Aa_0
\end{align*}
but from \ref{scintegral} we find
\begin{align*}
	z(\zeta) - 2A\zeta^{\frac{1}{2}} &= A\int_{x_1}^\zeta \frac{\dif t (t-x_3)}{\sqrt{(t-x_1)(t-x_2)(t-x_4)}} - A\int_0^\zeta \frac{\dif t}{\sqrt{t}} \\
	&= A\int_{x_1}^\zeta \frac{\dif t (t-x_3)}{\sqrt{(t-x_1)(t-x_2)(t-x_4)}} - A\int_{x_1}^{\zeta+x_1} \frac{\dif t}{\sqrt{t-x_1}} \\
	2Aa_0 &= \lim_{\zeta\to\infty}\bigg(z(\zeta) - 2A\zeta^{\frac{1}{2}} \bigg)\\ &= A\int_{x_1}^\infty \dif t \bigg(\frac{(t-x_3)}{\sqrt{(t-x_1)(t-x_2)(t-x_4)}} - \frac{1}{\sqrt{t-x_1}} \bigg).
\end{align*}
This contour integral is shown to be $0$ in \cite{hoppe1989quantum}. 

\subsection{Solution in terms of Elliptic Integrals}

We now introduce the (complete) elliptic integrals of the first and second kind, $K(\alpha)$ and $E(\alpha)$ respectively:
\begin{align*}
	K(\alpha) \coloneqq \int_0^{\frac{\pi}{2}} \frac{\dif\theta}{\sqrt{1-\alpha\sin^2\theta}} \qquad E(\alpha)\coloneqq \int_0^{\frac{\pi}{2}} \dif\theta \sqrt{1-\alpha\sin^2\theta} \qquad \vartheta(\alpha) \coloneqq \frac{K(\alpha)}{E(\alpha)}.
\end{align*}
For a reference on elliptic integrals, see for example \cite{lawden_elliptic}. Using the notation and structure of the argument in \cite{kazakov1999d} we define
\begin{align*}
	\lambda_i = \frac{x_i}{x_1-x_4} \qquad\qquad
	\alpha = \frac{x_2-x_4}{x_1-x_4} \qquad\qquad \alpha' = 1 - \frac{1}{\alpha} = \frac{x_2-x_1}{x_2-x_4}.
\end{align*}
Recall that $x_1 > x_2 > x_3 > x_4$, hence $0 < \alpha < 1$. The goal is to express $g_2$ and any quantities of interest (such as $\mu_2$), purely in terms of $\alpha$; if this is done, then fixing $\alpha$ is equivalent to fixing $g_2$, and consequently solves the model. \\

To this end, we note some basic properties of $K$ and $E$. The elliptic integrals at $\alpha$ and $\alpha'$ are related,
\begin{align*}
	K(\alpha') = \sqrt{\alpha}K(1-\alpha) \qquad \qquad E(\alpha') = \frac{1}{\sqrt{\alpha}}E(1-\alpha).
\end{align*}
Indeed,
\begin{align*}
	K(\alpha') &= \int_0^{\frac{\pi}{2}} \frac{\dif\theta}{\frac{1}{\sqrt{\alpha}}\sqrt{\alpha-(\alpha-1)\sin^2\theta}} \qquad\text{Sub: } t = \sin^2\theta \\
	&= \sqrt{\alpha}\int_0^1 \frac{\dif t}{2\sqrt{t(1-t)(\alpha-(\alpha-1)t)}} \qquad \text{Sub: } s = 1-t \\
	&= \sqrt{\alpha}\int_0^1 \frac{\dif s}{2\sqrt{(1-s)s(1-(1-\alpha)s)}} \qquad \text{Sub: } s = \sin^2\theta \\
	&= \sqrt{\alpha}K(1-\alpha)
\end{align*}
and an analogous argument holds for $E$. We also recall \textit{Legendre's relation}, which we do not prove here:
\[ E(\alpha)K(1-\alpha) + E(1-\alpha)K(\alpha)-K(\alpha)K(1-\alpha) = \frac{\pi}{2}. \]

We now rewrite the constraints \ref{c1},\ref{c2} and \ref{c3} in terms of elliptic integrals. Sum \ref{c2} and \ref{c3} to get
\begin{align}
	0 &= \int_{x_4}^{x_2} \frac{\dif t (t-x_3)}{\sqrt{(x_1-t)(x_2 - t)(t-x_4)}} \nonumber \\
	&= \int_0^{x_2-x_4} \frac{\dif t (x_3-x_4-t)}{\sqrt{(x_2-x_4-t)(x_1-x_4-t)t}} \qquad \text{Sub: } t' = \frac{t}{x_2-x_4} \nonumber \\
	&= \int_0^1 \frac{\dif t' [(x_3-x_4 - (x_1-x_4)\alpha t') ]}{\sqrt{(1-t')(1-\alpha t')t'}} \qquad\text{Sub: } t' = \sin^2\theta \nonumber \\
	&= \int_0^{\frac{\pi}{2}} \frac{\dif\theta [x_3-x_4 - (x_1-x_4)\alpha\sin^2\theta]}{\sqrt{1-\alpha\sin^2\theta}} \nonumber \\
	&= (x_3-x_1)K(\alpha) + (x_1-x_4)E(\alpha) \nonumber\\
	&= \lambda_2 K(\alpha) - K(\alpha) + 2E(\alpha) \nonumber\\
	\lambda_2 &= 1-2\vartheta(\alpha). \label{ellipticeq2}
\end{align}

Starting with \ref{c1}, we have
\begin{align}
	\sqrt{g_2} &= \int_{x_2}^{x_1} \frac{\dif t (t-x_3)}{(x_1-t)(t-x_2)(t-x_4)} \nonumber \\
	&= \int_0^{x_1-x_2}  \frac{\dif t(t+x_2-x_3)}{\sqrt{t(x_1-x_2-t)(t+x_2-x_4)}} \qquad \text{Sub: } t' = \frac{t}{x_1-x_2}\nonumber\\
	&= \int_0^1 \frac{\dif t' [(x_1-x_2)t'+(x_2-x_3)] }{\sqrt{t'(1-t')(x_2-x_4)(1-\alpha't')}} \qquad \text{Sub: } t' = \sin^2\theta \nonumber\\
	&= \frac{2}{\sqrt{x_2-x_4}}\int_0^{\frac{\pi}{2}} \frac{\dif\theta[(x_1-x_2)\sin^2\theta + (x_2-x_3)] }{\sqrt{1-\alpha'\sin^2\theta}} \nonumber\\
	&= 2\sqrt{x_2-x_4} E(\alpha') + \frac{2(x_4-x_3)}{\sqrt{x_2-x_4}}K(\alpha') \nonumber\\
	& \sqrt{g_2(x_2-x_4)} = 2(x_2-x_4)E(\alpha') - (x_1 + x_2 - x_4)K(\alpha'). \label{ellipticeq1}
\end{align}
Making use of Legendre's relation as well as $K(\alpha') = \sqrt{\alpha}K(1-\alpha)$ and $E(\alpha')=\frac{1}{\sqrt{\alpha}}E(1-\alpha)$, we further simplify \ref{ellipticeq1},
\begin{align*}
	\sqrt{g_2(x_2-x_4)}K(\alpha) &= \frac{2(x_2-x_4)}{\sqrt{\alpha}}E(1-\alpha)K(\alpha) - (x_1 + x_2 - x_4)\sqrt{\alpha}K(\alpha)K(1-\alpha) \\
	\sqrt{\frac{\alpha g_2}{x_2-x_4}}K(\alpha) &= 2E(1-\alpha)K(\alpha) - \frac{(x_1+x_2-x_4)}{x_1-x_4}K(\alpha)K(1-\alpha) \\
	\sqrt{\frac{g_2}{x_1-x_4}}K(\alpha)&= 2\bigg( \frac{\pi}{2} +K(\alpha)K(1-\alpha)-E(\alpha)K(1-\alpha) \bigg)\\&\qquad- (\lambda_2+1)K(\alpha)K(1-\alpha),
\end{align*}
and using $E(\alpha) = \frac{1}{2}(1-\lambda_2)K(\alpha)$ from \ref{ellipticeq2}, we get
\begin{align}
	\sqrt{\frac{g_2}{x_1-x_4}}K(\alpha) &= \pi +K(1-\alpha)(2K(\alpha)-(1-\lambda_2)K(\alpha))\nonumber\\&\qquad- (\lambda_2+1)K(\alpha)K(1-\alpha) = \pi \nonumber\\
	(x_1-x_4) &= \frac{g_2K(\alpha)^2}{\pi^2}.\label{ellipticeq3}
\end{align}

We now combine \ref{order-1/2} and \ref{order-3/2} in order to eliminate $(x_1-x_4)$ and solve for $g_2$ purely in terms of $\alpha$,
\begin{align*}
	12g_2 &= 2(x_1^2 + x_2^2 + x_4)^2 - (x_1+x_2+x_4)^2 = (x_1^2 + x_2^2 + x_4^2) - 2(x_1x_2 + x_1x_4 + x_2x_4) \\
	  &= (x_1-x_4)^2 + x_2^2 - 2x_1x_2 - 2x_2x_4 
	= (x_1-x_4)^2\bigg( 1 + \lambda_2^2 -2\lambda_1\lambda_2 - 2\lambda_2\lambda_4 \bigg).
\end{align*}
Now, note that
\begin{align*}
	\lambda_1 &= \frac{x_1}{x_1-x_4} = \frac{x_1-x_4}{x_1-x_4} - \frac{x_2 - x_4}{x_1-x_4} + \frac{x_2}{x_1-x_4} = 1 - \alpha + \lambda_2 \\
	\lambda_4 &= \frac{x_4}{x_1-x_4} = \frac{x_4-x_2}{x_1-x_4} + \frac{x_2}{x_1-x_4} = \lambda_2 - \alpha.
\end{align*}
Therefore, together with $\lambda_2 = 1-2\vartheta(\alpha)$ from \ref{ellipticeq2}, we are able to solve for $g_2$,
\begin{align*}
	6g_2 &= (x_1-x_4)^2\bigg(1 +\lambda_2(\lambda_2 - 2(1-\alpha+\lambda_2)-2(\lambda_2-\alpha))\bigg)= 2(x_1-x_4)^2\bigg(-3\vartheta^2+2(2-\alpha)\vartheta -(1-\alpha)\bigg) 
\end{align*}
where we have written $\vartheta\equiv\vartheta(\alpha)$. Finally, we combine the above with \ref{ellipticeq3}, again writing $K\equiv K(\alpha)$, 
\begin{align}
	\boxed{\frac{1}{g_2} = \frac{K^4}{3\pi^4}\bigg(-3\vartheta^2 + 2(2-\alpha)\vartheta - (1-\alpha)\bigg)} \label{g2solution}
\end{align}
hence solving for $g_2$ in terms of $\alpha$. \\

One can then solve for moments by equating higher order coefficients in the two expansions of $z(\zeta)$. We now return to \ref{order-5/2} and solve for $\mu_2$. Recalling that $\delta = 3\mu_2 - \frac{1}{4}$, equation \ref{order-5/2} reads
\begin{align*}
	\mu_2 &= \frac{1}{12} - \frac{2}{5}\frac{\gamma_3+x_3\gamma_2}{g_2^2}.
\end{align*}
We can express $\gamma_3 + x_3\gamma_2$ in terms of the following symmetric polynomials in $x_1,x_2,x_4$,
\begin{align*}
	P &= x_1x_2x_4 \\
	Q &= x_1^2x_2 + x_1^2x_4 + x_1x_2^2 + x_2^2x_4 + x_1x_4^2 + x_2x_4^2 \\
	R &= x_1^3 + x_2^3 + x_4^3
\end{align*}
so that 
\begin{align*}
	\gamma_3 &= {-\frac{1}{2}\choose 1}^3 P + {-\frac{1}{2} \choose 2}{-\frac{1}{2}\choose 1}Q + {-\frac{1}{2}\choose 3} R \\
	&= -\frac{1}{16}(2P+3Q+5R) \\
	x_3\gamma_2 &= \frac{1}{2}(x_1+x_2+x_4)\bigg( {-\frac{1}{2}\choose 1}^2(x_1x_2+x_1x_4 + x_2x_4) +{-\frac{1}{2}\choose 2}(x_1^2 + x_2^2 +x_4^2) \bigg) \\
	&= \frac{1}{16}( 6P + 5Q + 3R ).
\end{align*}
Therefore we have
\begin{align}
	\mu_2 = \frac{1}{12} - \frac{2}{5g_2^2}\frac{1}{16}(4P + 2Q - 2R) = \frac{1}{12}+\frac{1}{20g_2^2}(R-Q-2P) \label{mu2pqr}
\end{align}
By writing 
\[ R-Q-2P = (x_1-x_4)^3 (R'-Q'-2P') \]
where $P',Q',R'$ are the corresponding symmetric polynomials in $\lambda_1,\lambda_2,\lambda_4$, recalling that $\lambda_i = x_i/(x_1-x_4)$. By using the relationships 
\[ \lambda_1 = 1-\alpha+\lambda_2 \qquad\qquad \lambda_4 = \lambda_2 - \alpha \]
we can express each $P',Q',R'$ purely in terms of $\lambda_2$. Doing the algebra, we find
\begin{align*}
	P' &= \lambda_2^3 + (1-2\alpha)\lambda_2^2 + (\alpha^2-\alpha)\lambda_2\\
	Q' &= 6\lambda_2^3 + (6-12\alpha)\lambda_2^2 + (8\alpha^2 - 8\alpha + 2)\lambda_2 + (-2\alpha^3 + 3\alpha^2-\alpha) \\
	R' &= 3\lambda_2^3 + (-6\alpha+3)\lambda_2^2 +(6\alpha^2-6\alpha+3)\lambda_2 + (-2\alpha^3 + 3\alpha^2 - 3\alpha + 1) 
\end{align*}
therefore we find
\begin{align*}
	R'-Q'-2P' = -5\lambda_2^3 + (10\alpha-5)\lambda_2^2 + (-4\alpha^2+4\alpha+1)\lambda_2 + (-2\alpha+1).
\end{align*}
By using $\lambda_2 = 1-2\vartheta(\alpha)$, expanding, and expressing in terms of $\vartheta(\alpha)\equiv\vartheta$, we have
\begin{align*}
	R'-Q'-2P' &= 40\vartheta^3 + (40\alpha-80)\vartheta^2 + (8\alpha^2 - 48\alpha + 48)\vartheta + (-4\alpha^2+12\alpha-8).
\end{align*}
Finally, returning to \ref{mu2pqr}, we can express $\mu_2$ purely in terms of $\vartheta$, by using equations \ref{ellipticeq3} and \ref{g2solution},
\begin{align*}
	\mu_2 &= \frac{1}{12} + \frac{(x_1-x_4)^3}{20g_2^2}(R'-Q'-2P') \\
	&= \frac{1}{12} + \frac{g_2^3 K^6}{20g_2^2\pi^6}(R'-Q'-2P') \\
	&= \frac{1}{12} + \frac{K^6}{20\pi^6}\frac{3\pi^4}{K(\alpha)^4(-3\vartheta^2+(4-2\alpha)\vartheta+(\alpha-1))}(R'-Q'-2P') 
\end{align*}
and finally we arrive at the solution for $\mu_2$ in terms of $\alpha$,
\begin{align}
	\boxed{\mu_2 = \frac{1}{12} - \frac{3K^2}{5\pi^2}\frac{10\vartheta^3 + (10\alpha-20)\vartheta^2+(2\alpha^2-12\alpha+12)\vartheta -(\alpha^2-3\alpha+2)}{3\vartheta^2-(4-2\alpha)\vartheta-(\alpha-1))}} \label{mu2solution}
\end{align}
which agrees with \cite{kazakov1999d} (except for the typo which does not include the factor of 3).

\subsection{Analysis of the Fermionic Function}\label{fermionanalysis}

We proceed along the same lines as in \ref{bosonanalysis}, but this time we start with \ref{fermionfn}. We have
\begin{align*}
	z' &= \frac{A(\zeta - y_2)(\zeta - y_5)}{\zeta^2} \frac{1}{(1-y_1/\zeta)^\frac{1}{2}}\frac{1}{(1-y_3/\zeta)^\frac{1}{2}}\frac{1}{(1-y_4/\zeta)^{\frac{1}{2}}}\frac{1}{(1-y_6\zeta)^{\frac{1}{2}}} \\
	&= A(1- (y_2+y_5)\zeta^{-1}+y_2y_5\zeta^{-2})\prod_{k_1,k_3,k_4,k_6} \sum_{k_i=0}^\infty {-\frac{1}{2}\choose k_i}(-x_i\zeta^{-1})^{k_i} \\
	&= A(1- (y_2+y_5)\zeta^{-1}+y_2y_5\zeta^{-2})\sum_{k=0}^\infty (-1)^k\zeta^{-k}\\&\qquad \cdot \underbrace{\sum_{p+q+r+s=k} {-\frac{1}{2}\choose p}{-\frac{1}{2}\choose q}{-\frac{1}{2}\choose r}{-\frac{1}{2}\choose s}y_1^py_3^qy_4^ry_6^s }_{\eqqcolon \gamma_k}\\
	&= A\sum_{k=0}^\infty (-1)^k\gamma_k(\zeta^{-k}-(y_2+y_5)\zeta^{-k-1}+y_2y_5\zeta^{-k-2}) \\
	&= A\bigg( 1 - (\gamma_1 + y_2 + y_5)\zeta^{-1} +\sum_{k=2}^\infty (-1)^k(\gamma_k + (y_2+y_5)\gamma_{k-1} +y_2y_5\gamma_{k-2})\zeta^{-k}\bigg).
\end{align*}
Integration yields
\begin{align}
	z&\sim A\bigg(\zeta + a_0 - (\gamma_1+y_2+y_5)\ln \zeta +\sum_{k=1}^\infty \frac{(-1)^{k}}{k}(\gamma_{k+1} + (y_2+y_5)\gamma_{k} +y_2y_5\gamma_{k-1})\zeta^{-k}  \bigg).\label{fermionz}
\end{align}

We now compute the asymptotics of $F(z)=g_2z+i[W(iz+i/2)+W(iz-i/2)]$, first by using \ref{Wexpansion} to find
\begin{align*}
	W(iz\pm i/2) &= -i\frac{1}{z} \pm \bigg(\frac{i}{2}\bigg)\frac{1}{z^2} + i \bigg(\mu_2 - \frac{1}{4}\bigg)\frac{1}{z^3} \pm\frac{i}{2}\bigg( \frac{1}{4}-3\mu_2 \bigg)\frac{1}{z^4} + \mcal{O}(z^{-5}).
\end{align*}
Hence we have the following expansion for $F(z)$,
\begin{align}
	F(z) = g_2z +  \frac{2}{z} + \bigg(\frac{1}{2}-2\mu_2\bigg)\frac{1}{z^3} + \mcal{O}(z^{-5}). \label{Fasymptotics}
\end{align}
    Let $\delta=\frac{1}{2}-2\mu_2$. To use the Lagrange inversion, we put
\begin{align*}
	\frac{1}{F(z)} &\sim \frac{1}{g_z z} \bigg( \frac{1}{1 + \frac{2}{g_2 z^2} + \frac{\delta}{g_2z^4} + \cdots}  \bigg) \\
	&\sim \frac{1}{g_2}\frac{1}{z} - \frac{2}{g_2^2}\frac{1}{z^3} + \bigg( \frac{4}{g_2^3}-\frac{\delta}{g_2^2} \bigg)\frac{1}{z^5} + \mcal{O}(z^{-7}),
\end{align*}
and defining $f(w) = 1/F(1/w)$, we find that
\begin{align*}
	f(w) = \frac{w}{g_2} - \frac{2}{g_2^2}w^3 + \bigg( \frac{4}{g_2^3}-\frac{\delta}{g_2^2} \bigg)w^5.
\end{align*}
We can now apply Lagrange inversion to get the coefficients for $g$, the inverse of $f$
\begin{align*}
	[w^k]g = \frac{1}{k} \res_{w=0} F(1/w)^k.
\end{align*} 
Using \ref{Fasymptotics} we have $[w^{2k}]g=0$ and 
\begin{align*}
	[w]g &= g_2 \\ 
	[w^3] &= 2g_2^2\\
	[w^5] &= g_2^4\delta + 8g_2^3
\end{align*}
hence $g(z)\sim g_2z + 2g_2^2 z^3 + (g_2^4 \delta + 8g_2^3)z^5 + \cdots$ and we may now use $z = 1/g(1/F(z))=1/g(1/\zeta)$ to find
\begin{align*}
	z &\sim \frac{\zeta}{g_2}(1+2g_2\zeta^{-2} + (g_2^3\delta + 8g_2^2)+\cdots)^{-1} \\
	&\sim \frac{1}{g_2}\zeta - 2\frac{1}{\zeta} - (4g_2 + \delta g_2^2)\frac{1}{\zeta^3}+\mcal{O}(\zeta^{-5}).
\end{align*}
Therefore, we establish the following equations by comparing with \ref{fermionz},
\begin{align}
	\mcal{O}(\zeta^1):\qquad& \boxed{A = \frac{1}{g_2}} \\
	\mcal{O}(\ln\zeta):\qquad& \boxed{y_2 + y_5=-\gamma_1} \label{fermionlnzeta}\\  
	\mcal{O}(\zeta^0):\qquad& \boxed{a_0=0} \\
	\mcal{O}(\zeta^{-1}):\qquad&\boxed{y_2y_5 = 2 + \gamma_1^2 - \gamma_2 }\label{fermionzeta1/2}\\
	\mcal{O}(\zeta^{-2}):\qquad& \boxed{\gamma_3-2\gamma_1\gamma_2  +\gamma_1^3 + 2\gamma_1 = 0} \\
	\mcal{O}(\zeta^{-3}):\qquad&\boxed{\delta = \frac{1}{3g_2^2}(\gamma_4 - \gamma_3\gamma_1 - \gamma_2^2 + \gamma_2\gamma_1^2 + 2\gamma_2)+\frac{4}{g_2}}
\end{align}

Let us take stock of the situation we are in. We have parameters $y_1>\cdots>y_6$ and $A$, determining the Schwarz-Christoffel integral, $a_0$ is a constant of integration. We want to solve for $\mu_2$, or in principle higher moments if we were to go deeper into the expansion. We also have the integral constraints \ref{fermionconstraint1},\ref{fermionconstraint2},\ref{fermionconstraint3}. Therefore we do have enough equations to solve for $\mu_2$ or higher moments. However, the above equations are significantly more complex than in the Bosonic case, and having $6$ parameters $y_i$ appearing in the constraints \ref{fermionconstraint1},\ref{fermionconstraint2},\ref{fermionconstraint3} also further complicates matters. One can arrange to eliminate, e.g. $y_2$ and $y_5$ from the integral constraints using \ref{fermionlnzeta} and \ref{fermionzeta1/2}, but it still appears that one would need a new elliptic parameter in addition to the analogous parameter for $\alpha$ appearing in the bosonic case. One would then also need complete elliptic integrals of the third kind $\Pi(m,n)$ to fully express these equations in terms of elliptic integrals, and then solve the resulting equations for $\mu_2$. We have not found a way to explicitly carry out this procedure, however we note that the above system of equations is enough to allow numerical analysis of $\mu_2$ and, in principle, higher order moments.

\section*{Summary and outlook}

In this paper we studied type $(0,1)$ random fuzzy geometries coupled to bosonic and fermionic matter as bi-tracial matrix integral models. Such models provide a mathematically rigorous framework for studying quantum fluctuations of Dirac operators and can be viewed as finite-dimensional analogues of path integrals in noncommutative geometry and quantum gravity. The presence of matter fields introduces determinant-type interactions in the partition function, leading to nontrivial modifications of the associated matrix ensembles and their large $N$ behavior. 

Starting from the saddle point equation for the equilibrium measure, we derived the Schwinger--Dyson equations for arbitrary polynomial potentials using complex analytic methods. our approach relied on the construction of an entire periodic function built from the resolvent and its shifted values, whose analytic properties lead directly to recursive relations among the moments of the equilibrium measure. Once derived we provide a systematic perturbative framework for solving the model when either bosonic or fermionic contributions are considered.

In the Gaussian case, we obtained explicit analytic solutions for the free energy and the second moment in terms of elliptic integrals. The bosonic model was shown to be closely related to the Hoppe model and the three-colour matrix model, while the fermionic model gives rise to a parallel analytic structure. The appearance of elliptic functions and conformal mapping techniques highlights deep connections between random fuzzy geometries, matrix models, and complex analysis.

These results provide further evidence that Dirac ensembles form a rich intersection of noncommutative geometry, random matrix theory, and mathematical physics. They also suggest several future directions, including the study of multi-cut solutions, higher-signature fuzzy geometries, topological recursion, and the incorporation of gauge fields and internal degrees of freedom into random noncommutative geometries.

\section*{Acknowledgements}
 The authors would like to thank the Natural Sciences and Engineering Research Council of Canada (NSERC) for financial support.

\printbibliography

\appendix
\end{document}